\documentclass[a4paper]{article}
\usepackage{geometry}
\usepackage{mathrsfs}
\usepackage{bm}
\usepackage{authblk}
\usepackage{xcolor}

\newif\ifsage\sagefalse

\ifsage

\else
\usepackage{amsmath}
\usepackage{amsthm}
\usepackage{amsfonts}
\usepackage{amssymb}
\fi

\newtheorem{theorem}{Theorem}[section]

\newcommand{\field}{\mathbb}

\newcommand{\reals}{\field{R}}

\newcommand{\ie}{\textit{i.e.}~}
\newcommand{\eg}{\textit{e.g.}~}

\newcommand{\mbs}[1]{\bm{#1}}
\newcommand{\dexp}{\mathrm{dexp}}
\renewcommand{\skew}{\mathrm{skew}}

\newcommand{\nodrill}{\mbs{\chi}}
\newcommand{\pmap}[1]{{#1}_\perp}
\newcommand{\dmap}[1]{{#1}_\parallel}
\newcommand{\rotation}{{\mbs{{\Lambda}}}}

\newcommand{\dd}[2]{\frac{\mathrm{d}#1}{\mathrm{d}#2}}
\newcommand{\pd}[2]{\frac{\partial #1}{\partial #2}}

\newcommand{\ui}{\bm{E}_1}
\newcommand{\uj}{\bm{E}_2}
\newcommand{\uk}{\bm{E}_3}
\newcommand{\Order}[2]{\mathcal{O}(#1^{#2})}

\let\oldOmega=\Omega \renewcommand{\Omega}{\mathit{\oldOmega}}
\let\oldGamma=\Gamma \renewcommand{\Gamma}{\mathit{\oldGamma}}
\let\oldLambda=\Lambda \renewcommand{\Lambda}{\mathit{\oldLambda}}
\let\oldSigma=\Sigma \renewcommand{\Sigma}{\mathit{\oldSigma}}
\let\oldPi=\Pi \renewcommand{\Pi}{\mathit{\oldPi}}
\let\oldXi=\Xi \renewcommand{\Xi}{\mathit{\oldXi}}

\newcommand{\axial}{\bm{n}_\parallel}
\newcommand{\saxial}{n_\parallel}
\newcommand{\shear}{\bm{n}_\perp}
\newcommand{\torsion}{\bm{m}_\parallel}
\newcommand{\storsion}{m_\parallel}
\newcommand{\bartorsion}{\bar{m}_\parallel}
\newcommand{\tildetorsion}{\tilde{m}_\parallel}
\newcommand{\bending}{\bm{m}_\perp}
\newcommand{\Bending}{\bm{M}_\perp}

\begin{document}

\ifsage
\input{0_front}
\else

\title{Variational principles for nonlinear Kirchhoff rods}

\author[1,2]{Ignacio Romero}
\author[3]{Cristian Guillermo Gebhardt}
\affil[1]{IMDEA Materials Institute, C/ Eric Kandel 2, Tecnogetafe, Madrid 28906, Spain}
\affil[2]{Universidad Polit\'ecnica de Madrid, Jos\'e Guti\'errez Abascal, 2, Madrid 29006, Spain}
\affil[3]{Institute of Structural Analysis, Leibniz Universit\"at Hannover, Appelstra\ss e 9 A, 30167 Hannover, Germany}

\date{\today}

\maketitle

\begin{abstract}
  The present article studies variational principles for
  the formulation of static and dynamic problems involving
  Kirchhoff rods in a fully nonlinear setting. These results, some of them new, others
  scattered in the literature, are presented in a systematic
  way, helping to clarify certain aspects that have remained
  obscure. In particular,
  the study of transversely isotropic models reveals the
  delicate role that differential geometry plays in their
  formulation and unveils consequently some approximations that
  can be made to obtain simplified formulations.
\end{abstract}
\fi
\section{Introduction}
\label{sec-intro}
Since the seminal work of Kirchhoff in 1859 on rods subjected to bending and twisting, there have
been innumerable works devoted to the study of such models, starting from the classical
contributions by Clebsch \cite{Clebsch:1862wo}, the Cosserat brothers \cite{Cosserat:1909uf}, and
Love \cite{Love:1927wt}.  The topic, already fairly mature, is certainly rich and has motivated
detailed analyses during the last fifty years \cite{Antman:1974tc, Dill:1992wn, Coleman:1993hm, Langer:1996wj,
Antman:1995wm, Weiss:2002tq, OReilly:2017gs}.  Other recent contributions continue to
research more specific aspects of the theory like its stability \cite{Mielke:1988hk,
Goriely:2001eg}, the Hamiltonian structure \cite{Dichmann:1996fr, Singer:2008fw}, discrete models
\cite{Bergou:2008km}, relations with differential geometry \cite{Cheng:2017dk}, etc.

Despite being such a classical topic, there is a reborn interest in using and analyzing nonlinear
Kirchhoff rods, since they have found applications in fields as disparate  as DNA modeling
\cite{Benham:1979fm, Shi:1994ei, Schlick:1995km}, hair simulation \cite{Bertails:2006gd,
Kmoch:2009hm}, cables \cite{Coyne:1990hu, Boyer:2011hv}, M\"obius bands \cite{Moore:2018kn},
oil-well drill-strings \cite{Tucker:1999jp}, computer graphics \cite{Pai:2002bu}, space tethers
\cite{Valverde:2006cj}, vortex-jet filaments \cite{Fukumoto:2007}, knots \cite{Ivey:1999id,
Audoly:2007ko}, climbing plants \cite{McMillen:2002iz}, catheters for medical interventions
\cite{Wang:2017gk}, etc.  In many of these situations, numerical methods provide the only means for
arriving to useful solutions, motivating new developments in the computational side
\cite{Weiss:2002hh, Boyer:2004, Greco:2013co, Greco:2014hx, Meier:2016vz}.

In this work, we focus on variational formulations of Kirchhoff rods, and discuss in detail aspects
related to their restriction to the transversely isotropic case.  Our interest in the variational
approach is mainly due to the fact that it represents the basis of Galerkin-type numerical methods, most
notably finite elements. Motivated by the development of more efficient numerical methods, novel
Kirchhoff formulations have been proposed in recent years  \cite{Boyer:2011hv, Greco:2013co, Greco:2014hx,
Meier:2016vz} that demand a careful analysis, and studying their variational structure seems a
useful starting point. The latter provides a solid setting for the geometry of the configuration
space and a rigorous path for obtaining the governing equations of the problem, as well as the
consistent boundary and initial conditions. 

Variational principles for Kirchhoff rods are, of course, not new and many references
can be found where the governing equations of the rod are found from the stationarity conditions of
some potential energy \cite{Antman:1995wm, Langer:1996wj}, in the case of a quasistatic problem, or
Hamilton's action, in a dynamic context \cite{Dichmann:1996fr,Meier:2017gs}. The easiest way to formulate them
is to use variational principles for geometrically exact rods with shear deformation
\cite{Simo:1988uc,Simo:1985ur}, and then constrain the cross section orientation
to be orthogonal with respect to the tangent vector of the curve that describes the rod in the ambient space, following Kirchhoff's hypothesis. This results
in constrained optimization principles that are far from optimal from the computational point of view.

We review under which conditions the equations of a Kirchhoff rod can be obtained from a variational
principle without constraints. The key for such type of formulations is to eliminate the rotation
group from the configuration space of the rod \cite{Langer:1996wj, Boyer:2011hv, Greco:2013co, Meier:2014dk, Lefevre:2017ig}
and replace it with a simpler set that already accounts for the
Kirchhoff constraint.  It will be shown that it is indeed possible to formulate such models, and
provide the variational principle behind them, but only when the rod is transversely isotropic. This
is a very common situation both for analysis \cite{Langer:1996wj, Boyer:2011hv, Greco:2013co} as well
as applications \cite{Benham:1979fm, Audoly:2007ko, Boyer:2011hv, Lefevre:2017ig}, so its study is of
practical value. There are delicate assumptions that are sometimes implicitly made in this
transformation, assumptions that are not longer valid for dynamic problems. In the latter situation,
we show what other simplifications need to be done if the full rotation group is to be left out of
the formulation.

In view of the previous arguments, two types of variational principles are presented first: one
constrained principle for (static and dynamic) general Kirchhoff rods, and a second family of
principles for rods with isotropic cross sections. Finally, we will formulate a variational principle
that is valid for general rods but, when applied to transversely isotropic ones, decouples in a
certain way, making it very appealing for general purpose numerical applications.

The remainder of the article is structured as follows. Section~\ref{sec-unit-sphere} presents a
collection of results coming from differential geometry that will play a key role in the definition
of configuration spaces.  In Section~\ref{sec-formulation} we describe general concepts of nonlinear
rods, including deformation measures, stress resultants, and momenta. These concepts will be used
throughout the rest of the article. Section~\ref{sec-canonical} presents a constrained variational
principle for the static and dynamic analysis of extensible and inextensible Kirchhoff rods with
general cross sections. Section~\ref{sec-composite} departs from the main argument of the article to
introduce and discuss parametrizations of the rotational group as composite rotations. This is the
key geometrical argument behind rod models that eliminate the rotation group from their
configuration space, and hence we review it carefully. With these results at hand, we discuss in
Section~\ref{sec-ti} variational principles for transversely isotropic Kirchhoff rods, carefully
distinguishing between the quasistatic and dynamic cases, since there are different sets of
hypotheses in each case.  A variational principle suitable for general and transversely isotropic
rods is described in Section~\ref{sec-mixed}. Finally, the main results of the article are
summarized in Section~\ref{sec-summary}.

\section{Geometrical concepts}
\label{sec-unit-sphere}
Differential geometry plays a critical role in the description of nonlinear structural models, and
rods in particular. We review in this section the geometrical concepts employed in subsequent
sections.

\subsection{The unit sphere}
\label{subs-stwo}

The unit sphere $S^2$ plays a key role in the formulation of the rod models
discussed in this article. For completeness, and to clarify the notation, 
we summarize some of its main features, leaving to
other references more detailed descriptions~\cite{Eisenberg1979,RoUrreCy:2014,Romero:2017uv}.
This set is a nonlinear, smooth, compact, two-dimensional manifold defined as
\begin{equation}
  \label{eq-stwo}
  S^{2}
  :=
  \left\{
    \bm{d}\in\reals^{3}\mid\bm{d}\cdot\bm{d}=1
  \right\},
\end{equation}
where the dot product refers to the standard Euclidean inner product. In the context of structural
mechanics, especially rods and shells, elements of $S^2$ are sometimes referred to as
\emph{directors} \cite{Simo:1989wu,Romero:2004uo} since they are used in this type of models to
describe relevant geometrical directions.  Special attention must be paid to the fact that this
manifold possesses no special algebraic structure, specifically group-like structure.  

The tangent bundle of the unit 2-sphere is the manifold
\begin{equation}
  \label{eq-tstwo}
  TS^2 :=
  \left\{
    (\mbs{d},\mbs{c}),\
    \mbs{d}\in S^2,\
    \mbs{c}\in\reals^3,\
    \mbs{d}\cdot \mbs{c} = 0
  \right\} .
\end{equation}
Tangent vectors at a point $\mbs{d}\in S^2$ can be alternatively be described by
the relations
\begin{equation}
  \mbs{c} = \mbs{w}\times \mbs{d},
  \quad
  \mathrm{with}
  \quad
  \mbs{w}\cdot \mbs{d}= 0\ ,
\end{equation}
where the symbol ``$\times$'' denotes the cross product between vectors in $\reals^3$.

The unit 2-sphere, together with the metric inherited from $\reals^3$ has the structure of a
Riemannian manifold. By embedding the manifold in Euclidean three dimensional space, the covariant
derivative of a smooth vector field $\mbs{v}:S^2\to TS^2$ along a second vector field
$\mbs{w}:S^2\to TS^2$ is the vector field $\nabla_{\bm{w}}{\bm{v}}$, which evaluated at $\mbs{d}\in
S^2$ coincides with the projection of the derivative $D\mbs{v}$ in the direction of $\mbs{w}$ onto
the tangent plane to $\mbs{d}$.  If $\mbs{I}$ denotes the second order unit tensor and $\otimes$ the
dyadic product between vectors, this projection can be expressed as
\begin{equation}
  \nabla_{\mbs{w}} \mbs{v}
  :=
  (\mbs{I} - \mbs{d}\otimes \mbs{d})
  \;
  D \mbs{v} \cdot \mbs{w}\ .
\end{equation}
In particular, if $\mbs{d}:(a,b)\to S^2$ is a smooth one-parameter curve on the unit
sphere and $\mbs{d}'$ its derivative, the covariant derivative
of a smooth vector field $\mbs{v}:S^2\to TS^2$ in the direction of $\mbs{d}'$
can be evaluated as
\begin{equation}
  \nabla_{\mbs{d}'} \mbs{v}
  =
  (\mbs{I} - \mbs{d}\otimes \mbs{d})
  \;
  D \mbs{v} \cdot \mbs{d}'
  =
  (\mbs{v}\circ \mbs{d})'
  -
  \left(
    (\mbs{v}\circ \mbs{d})' \cdot \mbs{d}
  \right)
  \mbs{d}, 
\end{equation}
which, as before, is nothing but the projection of $(\mbs{v}\circ \mbs{d})'$
onto the tangent space $T_{\mbs{d}}S^2$.

\subsection{The special orthogonal group}

The set of proper orthogonal tensors also plays a prominent role in rod theories
and can be defined as
\begin{equation}
  SO(3) :=
  \left\{
    \mbs{\Lambda}\in \reals^{3\times3}, \
    \mbs{\Lambda}^T \mbs{\Lambda} = \mbs{\Lambda} \mbs{\Lambda}^T = \mbs{I}, \
    \det{\mbs{\Lambda}} = 1
    \right\} .
  \label{eq-sothree}
\end{equation}
This set possesses a group-like structure when considered with the tensor
multiplication operation, and it is also a smooth manifiold, hence
it is a Lie group. Its Lie algebra is 
the set
\begin{equation}
  so(3) :=
  \left\{
    \hat{\mbs{w}}\in\reals^{3\times3}, \
    \hat{\mbs{w}} = - \hat{\mbs{w}}^T
  \right\}.
  \label{eq-algebra}
\end{equation}
An isomorphism exists between vectors in $\reals^3$ and $so(3)$ defined as $\hat{\cdot}:\reals^3\to
so(3)$ such that for all $\mbs{w},\mbs{a}\in\reals^3$, the tensor $\hat{\mbs{w}}\in so(3)$ satisfies
$\hat{\mbs{w}}\mbs{a} = \mbs{w}\times \mbs{a}$.  The vector $\mbs{w}$ is referred to as the axial
vector of the skew-symmetric tensor $\hat{\mbs{w}}$ and we also write $\skew[\mbs{w}] = \hat{\mbs{w}}$.
The exponential map $\exp:so(3)\to SO(3)$ is a surjective application with a closed form expression given
by Rodrigues' formula
\begin{equation}
  \exp[\hat{\mbs{\theta}}]
  :=
  \mbs{I} +  \frac{\sin\theta}{\theta} \hat{\mbs{\theta}}
  +
  \frac{1}{2} \frac{\sin^2(\theta/2)}{(\theta/2)^2} \hat{\mbs{\theta}}^2\ ,
  \label{eq-rodrigues}
\end{equation}
with $\mbs{\theta}\in\reals^3$,  $\theta =|\mbs{\theta}|$,
and $|\cdot|$ denotes the Euclidean norm. The linearization
of the exponential map is simplified by introducing
the map $\dexp: so(3)\to\reals^{3\times 3}$ that
satisfies
\begin{equation}
  \dd{}{\epsilon}
  \exp[\hat{\mbs{\theta}}(\epsilon)]
  =
  \skew[\dexp[\hat{\mbs{\theta}}(\epsilon)] \mbs{\theta}'(\epsilon)]
  \exp[\hat{\mbs{\theta}}(\epsilon)]
  \label{eq-dexp}
\end{equation}
for every $\mbs{\theta}:\reals\to\reals^3$. Explicit expressions
of this map, and more aspects regarding the numerical treatment
of the rotation group can be found elsewhere \cite{Hairer:2002vg,Rom08,Romero:2017uv}.

\subsection{Composite rotations}
\label{subs-composite}
For any director~$\mbs{d}$, the three-dimensional Euclidean
space can be expressed as the direct sum
\begin{equation}
  \reals^3
  \cong   T_{\mbs{d}}S^2 \oplus \hbox{span}(\mbs{d})\ ,
  \label{eq-rthree-split}
\end{equation}
where $\hbox{span}(\mbs{d})$ is the linear subspace spanned by $\mbs{d}$. Given now two directors
$\mbs{d},\tilde{\mbs{d}}$, we say that a second order tensor $\mbs{T}:\reals^3\to\reals^3$
\emph{splits from $\mbs{d}$ to $\tilde{\mbs{d}}$} if it can be written in the form
\begin{equation}
  \mbs{T}= \pmap{\mbs{T}} + \dmap{\mbs{T}}\:,
  \label{eq-t-split}
\end{equation}
where $\pmap{\mbs{T}}$ is a bijection from $T_{\mbs{d}} S^2$
to $T_{\tilde{\mbs{d}}}S^2$ with $\ker(\pmap{\mbs{T}}) = \hbox{span}(\mbs{d})$,
and $\dmap{\mbs{T}}$ is a bijection from $\hbox{span}(\mbs{d})$ to $\hbox{span}(\tilde{\mbs{d}})$
with $\ker(\dmap{\mbs{T}}) = T_{\mbs{d}}S^2$. The
split~\eqref{eq-t-split} depends on the pair~$\mbs{d},\tilde{\mbs{d}}$
but it is not indicated explicitly in order to simplify the notation.

Let us now consider a one-parameter curve in $S^2$ denoted $\mbs{d}_t$
with $t\in[0,T]$. If $\mbs{d}$ is an arbitrary point on $\mbs{d}_t$,
the two dimensional space $T_{\mbs{d}}S^2$ can be viewed,
when $S^2$ is embedded in $\reals^3$,
as a tangent plane to the unit sphere at $\mbs{d}$.
Between $T_{\mbs{d}_0} S^2$ and $T_{\mbs{d}}S^2$
there exist, hence, infinite isomorphisms. For example, parallel transport
along $\mbs{d}_t$ provides a natural map between these two vector spaces. Another useful
transformation can be obtained from the unique rotation $\nodrill\in SO(3)$
that maps $\mbs{d}_0$ to $\mbs{d}$ \emph{without drill}, as long as
$\mbs{d}_0\ne -\mbs{d}$, and defined by
\begin{equation}
  \nodrill[\mbs{d}_0,\mbs{d}]
  :=
  (\mbs{d}_0\cdot \mbs{d}) \mbs{I}
  +
  \widehat{ \mbs{d}_0\times \mbs{d}}
  +
  \frac{1}{1+\mbs{d}_0\cdot \mbs{d}}
  (\mbs{d}_0\times \mbs{d})
  \otimes
  (\mbs{d}_0\times \mbs{d})\ .
  \label{eq-no-drill}
\end{equation}

The map $\nodrill[\mbs{d}_0,\mbs{d}]$ \emph{splits from $\mbs{d}_0$ to $\mbs{d}$}
and we can define
\begin{equation}
  \pmap{\nodrill}[\mbs{d}_0,\mbs{d}] = \nodrill - \mbs{d}\otimes \mbs{d}_0
  \ , \qquad
  \dmap{\nodrill}[\mbs{d}_0,\mbs{d}] = \mbs{d}\otimes \mbs{d}_0\ .
  \label{eq-nodrill-split}
\end{equation}
The map $\pmap{\nodrill}[\mbs{d}_0,\mbs{d}]$ is, by definition, 
a bijection between $T_{\mbs{d}_0} S^2$ and $T_{\mbs{d}}S^2$ which, in contrast with the one
defined by parallel transport, does not depend on the whole curve $\mbs{d}_t$ but only on the directors
$\mbs{d}_0$ and $\mbs{d}$.

Given a rotation $\nodrill$ that splits from $\mbs{d}_0$ to $\mbs{d}$ and any scalar
$\psi\in S^1$ (the unit circle), the map
\begin{equation}
  \rotation
  = \exp[\psi\, \hat{\mbs{d}}] \nodrill
  = \nodrill \exp[\psi\, \hat{\mbs{d}_0}]
  \label{eq-general-rotation}
\end{equation}
is also a rotation that splits from $\mbs{d}_0$ to $\mbs{d}$. Geometrically, $\rotation$
maps vectors on $T_{\mbs{d}_0}S^2$ to $T_{\mbs{d}}S^2$ with a drill angle equal
to~$\psi$.
This map will play a key role in the theory that follows.

\subsection{The Bishop (natural) frame and torsion}
\label{subs-bishop}

In the current context, it is desirable to describe the cross-section orientation along the
rod by means of an orthogonal frame that is somehow intrinsically related to the curve that
describes the rod in the ambient space.

To introduce such a rotation field, let $\bm{r}:\left[0,L\right]\rightarrow \reals^3$ be a
one-parameter curve with derivative~$\bm{r}'$.  Let us assume that this curve is regular, that is,
$\bm{r}'\ne\bm{0}$ everywhere on its domain. Additionally, we assume, without loss of generality,
that the curve is arc-length parametrized, \ie, $|\bm{r}'|=1$.  Now, let us consider
$\mbs{u},\mbs{v}:[0,L]\to S^2$ such that $\left\lbrace \bm{u}, \bm{v}, \bm{r}' \right\rbrace$ are mutually
orthogonal vector fields along the curve. We say that  $\left\lbrace \bm{u}, \bm{v}, \bm{r}'
\right\rbrace$ minimizes the rotation if the following conditions
\begin{equation}
  \bm{u}'\cdot\bm{v}=\bm{u}\cdot\bm{v}'=0
  \label{eq-mimimum-rotation}
\end{equation}
are satisfied. Let $\{\mbs{u}(s_0),\mbs{v}(s_0),\mbs{r}' (s_0)\} = \{ \mbs{u}_0,\mbs{v}_0,
\mbs{r}'_0\}$ be a known value of the orthonornal triad for some $s_0\in(0,L)$. Using this frame as
initial condition,  Eqs.~\eqref{eq-mimimum-rotation} can be integrated along $\mbs{r}$, defining
uniquely the vector fields $\mbs{u}$ and $\mbs{v}$ at all points of the curve. This curve has
Darboux vector
\begin{equation}
  \bm{k} = \bm{r}'\times\bm{r}''\:,
\end{equation}
which plays an essential role in defining parallel transport.

To analyze the concept of torsion in transported frames, let us consider the rotation field
$\exp[\psi\:\hat{\mbs{r}}']$, with $\psi:[0,L]\to S^1$, and the rotated triad
$\{\mbs{u}_\psi,\mbs{v}_\psi,\mbs{r}'\} = \exp[\psi\:\hat{\mbs{r}}'] \{\mbs{u},\mbs{v},\mbs{r}'\}$,
with
\begin{equation}
  \bm{u}_{\psi} = \cos(\psi)\bm{u}+\sin(\psi)\bm{v}
  \ ,\qquad
  \bm{v}_{\psi} =-\sin(\psi)\bm{u}+\cos(\psi)\bm{v}\:.
\end{equation}
The rotated frame has Darboux vector
\begin{equation}
	\bm{\omega} =\bm{k}+\psi'\bm{r}',
\end{equation}
or, equivalently,
\begin{equation}
\bm{\omega} = -(\bm{v}_{\psi}\cdot\bm{r}'')\bm{u}_{\psi}+(\bm{u}_{\psi}\cdot\bm{r}'')\bm{v}_{\psi}+\psi'\bm{r}'\:,
\label{eq-omega}
\end{equation}
where $\psi'$ is the torsion and $\psi$ is the torsion angle (for further details, see
\cite{AmericanMathematica:6j0BCmRM,Langer:1996wj}). The previous calculations show that
the frame $\{\mbs{u},\mbs{v},\mbs{r}'\}$ -- which is known as the natural or Bishop frame --
is the unique one obtained by transporting $\{\mbs{u}_0,\mbs{v}_0,\mbs{r}'_0\}$ along
the curve \emph{without torsion}.

Given, as before, a known frame $\{\mbs{u}_0,\mbs{v}_0,\mbs{r}'_0\}$ at a point $s_0\in(0,L)$, one
could transport it along the curve using the drill-free rotation given by
Eq.~\eqref{eq-no-drill}. Ignoring, for the moment, that this map is only defined as long as
$\mbs{r}'\ne \mbs{r}'_0$, we note that this alternatively transported frame has, in general, an induced torsion.  To see this, we observe that the Darboux vector emanated from the drill-free
rotation~$\nodrill[\mbs{r}'_0,\mbs{r}']$ is
\begin{equation}
  \mbs{\omega}_{\nodrill} = \mbs{r}'\times\mbs{r}''
  +\left(\mbs{a}\cdot\mbs{r}''\right)\mbs{r}'\ , 	 
\end{equation}
\noindent where
\begin{equation}
\mbs{a} =-\frac{\mbs{d}_0\times\mbs{r}'}{1+\mbs{d}_0\cdot \mbs{r}'}\ . 	
\end{equation}
The vector $\mbs{a}$, responsible for the torsion of the transported frame, is non-vanishing in
general, as anticipated.

More importantly, the term $\mbs{a}\cdot\mbs{r}''$ in the curvature $\mbs{\omega}_{\nodrill}$ is
non-integrable, meaning that there is no scalar function such that its derivative with respect to
arc-length produces torsion in the sense of Eq.~\eqref{eq-omega}.
As a result, the drill-free map cannot be corrected via a composite map~\eqref{eq-general-rotation}
yielding a frame that has no torsion. To verify
this assertion, we calculate the Darboux vector of this composite map to be
\begin{equation}
  \mbs{\omega}_{\rotation}
  =
  \mbs{r}'\times\mbs{r}''+\left( \psi'+\mbs{a}\cdot\mbs{r}''\right)\mbs{r}'\ . 	 
\end{equation}
If the map is to have vanishing torsion, the rotation angle $\psi$ would have
to satisfy
\begin{equation}
	\psi'+\mbs{a}\cdot\mbs{r}'' = 0\ ,
\end{equation}
or, explicitly,
\begin{equation}
  \psi(s) =
  -\int_{s_0}^s
  \mbs{a}\cdot \mbs{r}''
  \,\mathrm{d} \mu
  \ ,
\end{equation}
which clearly is a non-local quantity.
\section{Geometrically exact rods}
\label{sec-formulation}
We define in this section general concepts related to the kinematics of \emph{geometrically exact
rods}, that is, rod models for which no approximation whatsoever is made on the size of their
deformation. We describe first a very general model and then we indicate what constraints
can be imposed on it to obtain theories with reduced kinematics.

\subsection{General description}
\label{subs-general}
As customary, a rod is defined to be a three-dimensional deformable body whose length is much larger
than its other two dimensions. Such a body can be described by a curve in $\reals^3$ and a cross
section at every point of the curve whose intersection is precisely the barycenter of the
section. We will henceforth assume that these cross sections remain plane and undistorted at all
possible configurations of the rod. Models that account for the distorsion of the cross
section can be found elsewhere (\eg, \cite{Antman:1995wm}).

Let $L$ denote the length of the rod's centerline in the undeformed configuration, and $s\in[0,L]$
an arc-length coordinate employed to identify each point of the rod.  To describe the configuration
of the rod we select a fixed coordinate system with a Cartesian orthonormal basis
$\{\mbs{E}_i\}_{i=1}^3$. The position of the centerline point of arc-length coordinate~$s$ is
denoted $\mbs{r}(s)$. Moreover, since the cross section of the rod is assumed to remain plane and
undistorted at all configurations, two orthonormal vectors $\left\lbrace
\mbs{e}_1(s),\mbs{e}_2(s)\right\rbrace $ span the cross section with coordinate $s$ at all time, and we
assume that they are material vectors. Defining a third unit vector $\mbs{e}_3(s)=\mbs{e}_1(s)\times
\mbs{e}_2(s)$ we can uniquely identify the orientation of the cross section at $s$ by the tensor
$\mbs{\Lambda}(s)\in SO(3)$ such that
\begin{equation}
  \mbs{\Lambda}(s) \mbs{E}_i = \mbs{e}_i(s)\ .
  \label{eq-lambda-section}
\end{equation}

To make the presentation more precise, let us consider a rod that is clamped at
$s=0$. Based on the previous description, the configuration manifold of a
general rod of this type is the set
\begin{equation}
  Q :=
  \left\{
    (\mbs{r},\mbs{\Lambda}):[0,L]\to \reals^3\times SO(3),\ 
    \ \mbs{r}(0) = \bar{\mbs{r}}, \ \mbs{\Lambda}(0)= \bar{\mbs{\Lambda}}
  \right\},
  \label{eq-Q}
\end{equation}
where the functions $\mbs{r}$ and $\mbs{\Lambda}$ are assumed to be smooth
and $\bar{\mbs{r}},\bar{\mbs{\Lambda}}$ are known. Let us
stress here that this configuration space is not for Kirchhoff rods, but for
general models.

Among all possible configurations in $Q$, the reference configuration of the
rod is defined to be $(\mbs{r}_0,\mbs{\Lambda}_0)$, and we choose, without any loss
of generality, that the argument of these two functions coincides with the
arc-length of the curve defined by $\mbs{r}_0$, \ie,
\begin{equation}
  |\mbs{r}_0'(s)| = 1\ ,
  \label{eq-normalization}
\end{equation}
for all $s\in[0,L]$.

\subsection{Strain measures}
\label{subs-strains}
We study next all the possible ways in which a general rod, as defined in Section~\ref{subs-general},
can deform. More detailed discussions of the deformation modes of these rods
can be found in \cite{Simo:1985ur,Antman:1995wm}.

The first deformation measure is the axial strain, and gauges the local relative
value of the rod elongation, defined as
\begin{equation}
  \epsilon := (\mbs{r}'\cdot \mbs{e}_3) - 1\ .
  \label{eq-epsilon}
\end{equation}
The second measure is the shear strain, and it is defined as the vector
\begin{equation}
  \mbs{\sigma} := (\mbs{r}'\cdot \mbs{e}_\alpha)\; \mbs{e}_\alpha \quad \textrm{with} \quad \alpha = 1, 2\ .
  \label{eq-sigma}
\end{equation}
The relative change in the normal to the cross section is accounted for
by the bending strain
\begin{equation}
  \mbs{\kappa} := \mbs{e}_3 \times \mbs{e}_3'\ .
  \label{eq-bending}
\end{equation}
The last measure is the torsional strain defined as the scalar
\begin{equation}
  \tau := \mbs{e}_2\cdot \mbs{e}_1'\  ,
  \label{eq-torsion-strain}
\end{equation}
and accounts for the relative rotation of the cross section about the director.
For convenience, these measures are often gathered in two vectors
\begin{equation}
  \mbs{\gamma} := \mbs{\sigma} + \epsilon\, \mbs{e}_3 \ ,
  \qquad
  \mbs{\omega} := \mbs{\kappa} + \tau\, \mbs{e}_3\ ,  
  \label{eq-strains-spatial}
\end{equation}
which can be related with the centerline vector and orientation rotation
through the relations
\begin{equation}
  \mbs{\gamma} = \mbs{r}' - \mbs{e}_3\ ,
  \qquad
  \hat{\mbs{\omega}} = \mbs{\Lambda}' \mbs{\Lambda}^T\ .
  \label{eq-strains-formulae}
\end{equation}
The vectors $\mbs{\gamma}$ and $\mbs{\omega}$ are spatial strain measures
whose convected counterparts are, respectively,
\begin{equation}
  \mbs{\Gamma} := \mbs{\Lambda}^T \mbs{\gamma} \ ,
  \qquad
  \mbs{\Omega} := \mbs{\Lambda}^T \mbs{\omega}\ .
  \label{eq-strains-convected}
\end{equation}
and from these, we define the convected shear and bending strains
\begin{equation}
  \mbs{\Sigma} := \mbs{\Lambda}^T \mbs{\sigma}\ ,
  \qquad
  \mbs{K} := \mbs{\Lambda}^T \mbs{k}\ ,
  \label{eq-convected-sk}
\end{equation}
respectively.
All of these measures are frame invariant under the (left) action of
the special Euclidean group $SE(3)$, \ie, the group of isometries on the configuration
space $Q$. See \cite{Simo:1985ur,Antman:1995wm}.

\subsection{Stored energy and stress resultants}
\label{subs-energy}
In this section we study the most general hyperelastic model
for a rod, and identify the stress resultants that are
work-conjugate to the strains identified in Section~\ref{subs-strains}.

To start, let us postulate the existence of an objective stored energy function of
the form $U=U(\mbs{\sigma},\epsilon,\mbs{\kappa},\tau; s)$. The objectivity of this
function implies that
\begin{equation}
  U( \mbs{\sigma}, \epsilon, \mbs{\kappa}, \tau; s)
  =
  U( \mbs{Q}\mbs{\sigma}, \epsilon , \mbs{Q}\mbs{\kappa}, \tau; s)
\end{equation}
for every $\mbs{Q}\in SO(3)$. By selecting, in particular, $\mbs{Q}=\mbs{\Lambda}^T$, 
it follows that
\begin{equation}
  U(\mbs{\sigma}, \epsilon, \mbs{\kappa}, \tau; s) =
  \tilde{U}(\mbs{\Sigma}, \epsilon, \mbs{K}, \tau; s)
  .
  \label{eq-strain-energy}
\end{equation}
Hence, by expressing the stored energy in terms of convected measures --- which are
invariant under the action of the special Euclidean group --- we guarantee its objectivity.

The differential of the stored energy function can be calculated as
\begin{equation}
  \mathrm{d}\,\tilde{U} =
  \pd{\tilde{U}}{\mbs{\Sigma}} \cdot \mathrm{d}\, \mbs{\Sigma} +
  \pd{\tilde{U}}{\epsilon} \mathrm{d}\,{\epsilon} +
  \pd{\tilde{U}}{\mbs{K}}\cdot \mathrm{d}\,\mbs{K} +
  \pd{\tilde{U}}{\tau} \mathrm{d}\,\tau
  \ ,
  \label{eq-stress-power}
\end{equation}
and we define the convected stress resultants
\begin{equation}
  \mbs{\Xi} :=\pd{\tilde{U}}{\mbs{\Sigma}}\ ,
  \qquad
  \saxial := \pd{\tilde{U}}{\epsilon}\ ,
  \qquad
  \Bending := \pd{\tilde{U}}{\mbs{K}}\ ,
  \quad
  \storsion := \pd{\tilde{U}}{\tau}\ ,
  \label{eq-stress}
\end{equation}
which are work conjugate, respectively, to the shear, axial, bending, 
and torsional strain measures in their convected form. The spatial
stress resultants are the push-forwards $(\mbs{\xi}, \saxial, \bending, \storsion)
=(\mbs{\Lambda} \mbs{\Xi}, \saxial, \mbs{\Lambda}\Bending,\storsion)$.

For the remainder of this article, we will focus on beam models that have no shear deformation, \ie\
$\mbs{\Xi}=\mbs{0}$. In these situations, hence, it will be unnecessary to account for the
contribution to the internal energy due to shear and we will always assume that the latter is of the
form $\tilde{U}= \tilde{U}(\epsilon,\mbs{K},\tau; s)$.

\subsection{Momenta and kinetic energy}
\label{subs-momenta}
Let us now consider a motion of the rod, that is, a time-parameterized curve in configuration space
described by a pair of functions $(\mbs{r}(s,t), \mbs{\Lambda}(s,t) ) $ such that
$(\mbs{r}(\cdot,t),\mbs{\Lambda}(\cdot,t))\in Q$ for all $t\in[0,T]$.  Using  a superposed dot as the
notation for the derivative with respect to time, the generalized velocity of the rod is the vector field
$(\dot{\mbs{r}},\dot{\mbs{\Lambda}})$ belonging, for every $t\in[0,T]$, to the tangent bundle
\begin{equation}
  TQ :=
  \left\{(\mbs{r},\mbs{\Lambda}; \mbs{y}, \mbs{Y}):[0,L]  \to \reals^3\times SO(3) \times \reals^3\times TSO(3),\
    \mbs{y}(0) = \mbs{0}, \
    \mbs{Y}(0) = \mbs{0}
  \right\} .
  \label{eq-tq}
\end{equation}
The time derivative of the centroid position is the velocity
$\mbs{v}=\dot{\mbs{r}}$, and the derivative of the rotation can be written as
\begin{equation}
  \dot{\mbs{\Lambda}} = \hat{\mbs{w}} \mbs{\Lambda} = \mbs{\Lambda} \widehat{ \mbs{W}}\ ,
  \label{eq-omegas}
\end{equation}
where $\mbs{w}$ and $\mbs{W}$ are the spatial and convected angular
velocities, respectively. These fields can be split as in
\begin{equation}
  \mbs{w} = \mbs{w}_{\perp} + w_\parallel \mbs{e}_3\ ,
  \qquad
  \mbs{W} = \mbs{W}_\perp + W_\parallel \mbs{E}_3\ ,
  \label{eq-omegas-split}
\end{equation}
with
\begin{equation}
  \mbs{w}_{\perp} = \mbs{e}_3\times \dot{\mbs{e}_3}\ ,
  \qquad
  \mbs{W}_{\perp} = \mbs{\Lambda}^T \mbs{w}_{\perp}\ .
  \qquad
  \label{eq-omegas-perp}
\end{equation}
The kinetic energy density of the rod is defined as
\begin{equation}
  k :=
  \frac{1}{2} A_\rho |\mbs{v}|^2 +
  \frac{1}{2} \mbs{w} \cdot \mbs{i}_\rho \mbs{w}\ .
  \label{eq-kinetic-density}
\end{equation}
The constant $A_\rho$ is equal to the product of the cross section area
and the density of the rod's material. The tensor $\mbs{i}_\rho$ is
the spatial inertia of the cross section. The convected
inertia, constant in time, is the pullback of $\mbs{i}_\rho$ to the
reference configuration, namely,
\begin{equation}
  \mbs{I}_\rho :=
  \mbs{\Lambda}^T \mbs{i}_\rho \mbs{\Lambda}\ .
  \label{eq-convected-inertia}
\end{equation}
The convected inertia $\mbs{I}_\rho:\reals^3\to\reals^3$ is a second order, symmetric, positive
definite tensor that splits from $\mbs{E}_3$ to $\mbs{E}_3$, \ie, it is of the form
\begin{equation}
  \mbs{I}_\rho = \mbs{I}_\perp + I_\parallel\, \mbs{E}_3\otimes \mbs{E}_3\ ,
  \label{eq-irho-decomp}
\end{equation}
where $\mbs{I}_\perp$ maps bijectively $\mathrm{span}(\mbs{E}_1,\mbs{E}_2)$ onto
itself and satisfies $\mbs{I}_\perp \mbs{E}_3 =
\mbs{0}$. By pushing this inertia to the current configuration, we find that
the spatial inertia admits a similar split
\begin{equation}
  \mbs{i}_\rho = \mbs{i}_\perp + i_\parallel\, \mbs{e}_3\otimes \mbs{e}_3\ ,
  \label{eq-irho-decomp-spatial}
\end{equation}
where now $\mbs{i}_\perp$ splits from $\mbs{e}_3$ to $\mbs{e}_3$.  As a
consequence of the structure of the inertia tensor, the rotational part of the kinetic energy
density can be written in the following equivalent ways:
\begin{equation}
  \frac{1}{2} \mbs{w} \cdot \mbs{i}_\rho \mbs{w}
  =
  \frac{1}{2} \mbs{W} \cdot \mbs{I}_\rho \mbs{W}
  =
  \frac{1}{2} \mbs{W}_\perp \cdot \mbs{I}_\perp \mbs{W}_\perp
  +
  \frac{1}{2} I_\parallel\, W_\parallel^2
  =
  \frac{1}{2} \mbs{w}_\perp \cdot \mbs{i}_\perp \mbs{w}_\perp
  +
  \frac{1}{2} i_\parallel\, w_\parallel^2 .
  \label{eq-kinetic-density-convected}
\end{equation}

The translational and rotational momenta of the rod are conjugate
to the velocities as in
\begin{equation}
  \mbs{p}
  :=
  \pd{k}{\mbs{v}} = A_\rho \mbs{v}\:,
  \qquad
  \hbox{and}
  \qquad
  \mbs{\pi}
  :=
  \pd{k}{\mbs{w}} = \mbs{i}_\rho \mbs{w}\:,
  \label{eq-momenta}
\end{equation}
and we note that we can introduce a convected version of
the momentum $\mbs{\pi}$ by pulling it back, as before, and defining
\begin{equation}
  \mbs{\Pi} := \mbs{\Lambda}^T \mbs{\pi} = \pd{k}{\mbs{W}}  \ .
  \label{eq-pi-convected}
\end{equation}
Due to the particular structure of the inertia
the momenta can also be split, as before, as in
\begin{equation}
  \mbs{\pi} = \mbs{\pi}_\perp + \pi_\parallel \mbs{e}_3\ ,
  \qquad
  \mbs{\Pi} = \mbs{\Pi}_\perp + \Pi_\parallel \mbs{E}_3
  \label{eq-pi-split}
\end{equation}
with
\begin{equation}
  \mbs{\pi}_\perp = \mbs{i}_\perp \mbs{w}_\perp\ ,
  \qquad
  \mbs{\Pi}_\perp = \mbs{I}_\perp \mbs{W}_\perp ,
  \qquad
  \pi_\parallel = \Pi_\parallel = i_\parallel w_\parallel = I_\parallel W_\parallel \ .
  \label{eq-pi-split-components}
\end{equation}
%
\section{The canonical variational principle for Kirchhoff rods}
\label{sec-canonical}
In this section we obtain the governing equations of Kirchhoff rods from
Hamilton's principle of stationary action in the most straightforward way.

\subsection{Lagrangian}
\label{subs-k-lagrangian}
Kirchhoff rods are a subclass of the general rod models described in Section~\ref{sec-formulation},
with the condition of vanishing shear deformation, \ie, $\mbs{\Sigma=\mbs{0}}$. In view of Eq.~\eqref{eq-sigma},
this constraint is equivalent to having $\mbs{r}'$ parallel to
the vector~$\mbs{e}_3$, or equivalently
\begin{equation}
  \mbs{0}
  =
  (\mbs{\Lambda}^T \mbs{r}') \times \mbs{E}_3\ .
  \label{eq-constrain-shear}
\end{equation}
Instead of trying to define a configuration space for this class of rods, we will
work with the general space defined in Eq.~\eqref{eq-Q}, and constrain the
kinematics of the model using Lagrange multipliers to impose pointwise
relation~\eqref{eq-constrain-shear}.

To make the exposition as simple as possible, we assume that the rod is clamped
at the section corresponding to $s=0$. Moreover, a known
conservative distributed force, per unit of reference length,
$\bar{\mbs{n}}:[0,L]\to\mathbb{R}^3$ is applied along the rod, and a conservative point
force $\tilde{\mbs{n}}$ is applied at the free end, namely $s=L$. Under this loading,
the potential energy of the rod is 
\begin{equation}
  V = \int_0^L \tilde{U}(\epsilon, \mbs{K}, \tau; s) \,\mathrm{d} s
  -
  \int_0^L \bar{\mbs{n}}\cdot \mbs{r} \,\mathrm{d} s
  -
  \tilde{\mbs{n}}\cdot \mbs{r}(L)\ .
  \label{eq-kirchhoff-v}
\end{equation}
Using the definitions of Section~\ref{subs-momenta}, the kinetic energy
of this system can be calculated to be
\begin{equation}
  T =
  \int_0^L
  \left(
    \frac{1}{2} A_\rho |\dot{\mbs{r}}|^2
    +
    \frac{1}{2} \mbs{w}\cdot \mbs{i}_\rho \mbs{w}
  \right)
  \,\mathrm{d} s\ ,
  \label{eq-kirchhoff-t}
\end{equation}
and the Lagrangian is just $\mathcal{L}=T-V$.

\subsection{Variations of the strain measures and rates}
\label{subs-k-variations}
The stationarity conditions for the action will be obtained using
calculus of variations. We gather next some results that will prove
necessary for the computation of the functional derivatives and, later,
for the linearization of the model.

To introduce these concepts, let us consider a curve of configurations
$(\mbs{r}_\iota, \mbs{\Lambda}_\iota)$ parametrized by the scalar $\iota$ and given by
\begin{equation}
  \left(\mbs{r}_\iota(s,t),\mbs{\Lambda}_\iota(s,t)\right)
  =
  \left(
    \mbs{r}(s,t) + \iota\, \delta\mbs{r}(s,t) ,
    \exp[\iota\, \widehat{\delta\mbs{\theta}}(s,t)] \mbs{\Lambda}(s,t)
  \right)\ ,
  \label{eq-one-parameter}
\end{equation}
where $\delta \mbs{r}:[0,L]\times[0,T]\to\reals^3$ and
$\delta\mbs{\theta} :[0,L]\times[0,T]\to\reals^3$ are arbitrary variations
that satisfy
\begin{equation}
  \delta \mbs{r}(0,t) =
  \delta \mbs{r}(s,0) = \delta \mbs{r}(s,T) = \mbs{0}\ ,
  \qquad
  \delta \mbs{\theta}(0,t) =
  \delta \mbs{\theta}(s,0) = \delta \mbs{\theta}(s,T) = \mbs{0} \ .
  \label{eq-variations}
\end{equation}
The curve $(\mbs{r}_\iota,\mbs{\Lambda}_\iota)$ passes through the configuration
$(\mbs{r},\mbs{\Lambda})$ when $\iota=0$ and has tangent at this point equal to
\begin{equation}
  \left. \pd{}{\iota} \right|_{\iota=0}
  (\mbs{r}_\iota,\mbs{\Lambda}_\iota)
  =
  \left( \delta \mbs{r}, \widehat{\delta\mbs{\theta}} \mbs{\Lambda} \right) .
  \label{eq-tangent-variation}
\end{equation}
For future reference let us calculate the variation of the derivative $\mbs{\Lambda}'$.
To do so, let us first define the arc-length derivative of the perturbed rotation, that is,
\begin{equation}
  \begin{split}    
    \pd{}{s} \mbs{\Lambda}_\iota
    &=
    \pd{}{s} \exp[\iota \widehat{\delta\mbs{\theta}}] \mbs{\Lambda}
    =
    \skew\left[ \dexp[\iota \widehat{\delta \mbs{\theta}}] \iota \delta \mbs{\theta}'\right]
    \mbs{\Lambda} + \exp[\iota \widehat{\delta \mbs{\theta}}] \mbs{\Lambda}' \ .
  \end{split}
  \label{eq-d-lambdaprime}
\end{equation}
Then, the variation of $\mbs{\Lambda}'$ is just
\begin{equation}
    \delta ( \mbs{\Lambda}' )
    =
    \left.\pd{}{\iota}\right|_{\iota=0}
    \pd{}{s} \mbs{\Lambda}_\iota
    =
    \widehat{\delta \mbs{\theta}'} \mbs{\Lambda} + \widehat{\delta \mbs{\theta}} \mbs{\Lambda}'\ .
  \label{eq-var-lambdaprime}
\end{equation}
With the previous results we can now proceed to calculate the variations of
the strain measures, as summarized in the following theorem.

\begin{theorem}
  The linearization of the three strain measures $(\epsilon, \mbs{K}, \tau)$
  is
\begin{equation}
  \begin{split}
    \delta\epsilon
    &= \mbs{e}_3\cdot \delta \mbs{r}' + \mbs{e}_3\times \mbs{r}'\cdot\delta \mbs{\theta}
    \ ,
    \\
    \delta \mbs{K}
    &=
    \rotation^T
      (\mbs{I}-\mbs{e}_3 \otimes \mbs{e}_3) \delta \mbs{\theta}'
      \ ,
    \\
    \delta\tau &=
    \mbs{e}_3 \cdot \delta \mbs{\theta}'\ .
  \end{split}
   \label{eq-strain-linear}
\end{equation}
\end{theorem}

\begin{proof}
The strain measures of the one-parameter curve of configurations
$(\mbs{r}_\iota,\mbs{\Lambda}_\iota)$ are
\begin{equation}
  \epsilon_\iota = \mbs{r}_\iota'\cdot \mbs{e}_{3,\iota} - 1\ ,
  \qquad
  \mbs{K}_\iota = \mbs{E}_3 \times ( \mbs{\Lambda}^T_{\iota} \mbs{e}'_{3,\iota})\ ,
  \qquad
  \tau_\iota =  \mbs{e}_{2,\iota} \cdot  \mbs{e}'_{1,\iota} 
  \label{eq-str-proof-0}
\end{equation}
where $\mbs{e}_{i,\iota} = \mbs{\Lambda}_\iota \mbs{E}_i$. 
The variation of the axial strain measure is computed as follows:
\begin{equation}
  \delta \epsilon
  =
  \left. \pd{}{\iota} \right|_{\iota=0}
  \left( \mbs{r}'_\iota\cdot \mbs{e}_{3,\iota} - 1 \right)
  =
  \mbs{e}_3\cdot \delta \mbs{r}' + \mbs{e}_3\times \mbs{r}' \cdot \delta \mbs{\theta}\ .
  \label{eq-str-proof-1}
\end{equation}
The variation of the bending strain is obtained from its
definition employing some algebraic manipulations and
expression~\eqref{eq-var-lambdaprime} as follows:
\begin{equation}
  \begin{split}
  \delta \mbs{K}
  &=
  \left. \pd{}{\iota} \right|_{\iota=0}
  \left(
    \mbs{E}_3 \times ( \mbs{\Lambda}^T_\iota \mbs{e}'_{3,\iota})
  \right)
  \\
  &=
  \mbs{E}_3 \times
  \left(
    \delta \mbs{\Lambda}^T \mbs{\Lambda}' \mbs{E}_3 + \mbs{\Lambda}^T
    \delta \mbs{\Lambda}' \mbs{E}_3
  \right)
  \\
  &=
  \mbs{\Lambda}^T
  \left(
    \mbs{e}_3 \times
    \left( \widehat{\delta \mbs{\theta}'}\times \mbs{e}_3 \right)
    \right)
  \\
  &=
    \mbs{\Lambda}^T
    \left(
      \delta \mbs{\theta}' - ( \delta \mbs{\theta}' \cdot \mbs{e}_3) \mbs{e}_3
    \right)
    \\
  &=
    \mbs{\Lambda}^T
	(\mbs{I}-\mbs{e}_3 \otimes \mbs{e}_3) \delta \mbs{\theta}'    
  \ .
  \end{split}
  \label{eq-str-proof-2}
\end{equation}
The linearization of the torsional strain follows similar steps:
\begin{equation}
  \begin{split}
    \delta\tau
    &=
    \left. \pd{}{\iota} \right|_{\iota=0}
    \left( \mbs{e}_{2,\iota}\cdot \mbs{e}'_{1,\iota} \right)
    \\
    &=
    \widehat{\delta \mbs{\theta}} \mbs{\Lambda} \mbs{E}_2\cdot \mbs{\Lambda}' \mbs{E}_1
    +
    \mbs{\Lambda} \mbs{E}_2 \cdot
    \left(
      \widehat{\delta\mbs{\theta}'}\mbs{\Lambda} + \widehat{\delta\mbs{\theta}}\mbs{\Lambda}'
    \right)
    \mbs{E}_1
    \\
    &=
    \mbs{e}_1\times \mbs{e}_2 \cdot \delta \mbs{\theta}'
        \\
    &=
    \mbs{e}_3 \cdot \delta \mbs{\theta}'\ .    
  \end{split}
   \label{eq-str-proof-3}
\end{equation}
\end{proof}

The linearization of the rates is almost identical to the one of the
strains. For convenience, we present without proof the attendant results in the following theorem:

\begin{theorem}
  The linearization of the three time rates $(\mbs{v},\mbs{W}_\perp,w_\parallel)$ is:
  \begin{equation}
  \begin{split}
    \delta \mbs{v} &= \delta \dot{\mbs{r}} \ , \\
    \delta \mbs{W}_\perp
    &=
    \mbs{\Lambda}^T (\mbs{I}-\mbs{e}_3 \otimes \mbs{e}_3) \delta \dot{\mbs{\theta}} \ ,
    \\
    \delta w_\parallel
    &=
    \mbs{e}_3 \cdot \delta \dot{\mbs{\theta}}\ .
  \end{split}
   \label{eq-rates-linearization}
\end{equation}
\end{theorem}

\subsection{Governing equations}
\label{subs-kir-equations}

Hamilton's principle of stationary action states that the governing
equations of the clamped-free rod are the Euler-Lagrange equations of the
constrained action functional
\begin{equation}
  S =
  \int_0^T
  \left(
    T - V
    - \int_0^L \mbs{\eta}\cdot (\mbs{\Lambda}^T \mbs{r}')\times \mbs{E}_3
    \,\mathrm{d} s
  \right)
  \,\mathrm{d} t\ ,
  \label{eq-kirchhoff-action}
\end{equation}
with unknown fields $(\mbs{r}, \mbs{\Lambda}, \dot{\mbs{r}},\dot{\mbs{\Lambda}})$ in $TQ$ and
Lagrange multiplier $\mbs{\eta}:[0,L]\to \mathrm{span}(\mbs{E}_1,\mbs{E}_2)$. The main result of the section is as
follows:

\begin{theorem}
	\label{thm-k-g}
  The Euler-Lagrange equations of the action~\eqref{eq-kirchhoff-action} are:
  \begin{subequations}
    \begin{align}
      \mbs{n}' + \bar{\mbs{n}} &= \dot{\mbs{p}}\ , & \label{eq-el-k1-1} \\
      \mbs{m}' + \mbs{r}'\times \mbs{n} &= \dot{\mbs{\pi}}\ , & \label{eq-el-k1-2}\\
      (\mbs{\Lambda}^T \mbs{r}') \times \mbs{E}_3 &= \mbs{0}\ , & \label{eq-el-k1-3}
  \end{align}
   \label{eq-el-k1}
 \end{subequations}
 where the stress resultants are defined as
 \begin{equation}
   \begin{split}     
   \mbs{n} &= \shear + \axial\ , 
   \qquad
   \shear = (\mbs{\Lambda}\mbs{\eta})\times\mbs{e}_3\ ,
   \qquad
   \axial = \pd{\tilde{U}}{\epsilon} \mbs{e}_3\ ,
   \\
   \mbs{m} &= \bending + \torsion\ ,
   \qquad
   \bending = \mbs{\Lambda}\pd{\tilde{U}}{\mbs{K}}\ ,
   \qquad
   \torsion = \pd{\tilde{U}}{\tau} \mbs{e}_3\ .
   \end{split}
  \label{eq-el-k2}
\end{equation}
Eq.\eqref{eq-el-k1} can be rewritten as
\begin{subequations}
  \begin{align}
    \left( \axial + \frac{\mbs{d}}{|\mbs{r}'|}\times \nabla_{\mbs{d}'} \mbs{m} \right)'
    + \bar{\mbs{n}}
    &=
       \: \dot{\mbs{p}} +
       \left(
       \frac{\mbs{d}}{|\mbs{r}'|} \times \nabla_{\dot{\mbs{d}}} \mbs{\pi}\right)' \ ,  \label{eq-el-k3-1}
    \\
    \mbs{m}' \cdot \mbs{d}
     &= \dot{\mbs{\pi}} \cdot \mbs{d}\ ,
	 \label{eq-el-k3-2}
  \end{align}
   \label{eq-el-k3}
\end{subequations}
where $\mbs{d} = \mbs{e}_3 = \frac{\mbs{r}'}{|\mbs{r}'|}$. The natural boundary conditions at $s=L$ are
\begin{equation}
  \axial+\frac{\mbs{d}}{|\mbs{r}'|}\times \left(\nabla_{\mbs{d}'}
    \mbs{m}-\nabla_{\dot{\mbs{d}}}\mbs{\pi}\right) = \tilde{\mbs{n}}
  \ ,\qquad
  \mbs{d}\times\mbs{m} = \mbs{0}\ . 
	\label{eq-L-bc-k3}
\end{equation}

\end{theorem}

\begin{proof} The first part of the theorem follows from a systematic calculation of $\delta S$, the
variation of the action, and the results of Section~\ref{subs-strains}, and thus we omit a detailed derivation. It
should be noted that $\mbs{n}$ is the contact force on the cross section, which can be additively
decomposed on an axial part, $\axial$, and the shear force $\shear$, which appears as a result
of the constraint. The contact torque $\mbs{m}$ itself is the sum of the bending moment $\bending$
and the torsional moment $\torsion$.

Let $\mbs{d}= \mbs{r}'/|\mbs{r}'|$. Assuming $\mbs{r}$ is a smooth function in both $s$ and $t$,
and $|\mbs{r}'|>0$, condition \eqref{eq-el-k1-3} is equivalent to the constraint
$\mbs{d}=\mbs{e}_3$. Then, $\axial = \pd{\tilde{U}}{\epsilon} \mbs{d}$ and  $\mbs{r}'\times
\mbs{n}= \mbs{r}'\times \shear$. From Eq.~\eqref{eq-el-k1-2} if follows that
  \begin{equation}
	 \mbs{r}'\times \shear = \dot{\mbs{\pi}} - \mbs{m}'
  \label{eq-proof-kel-2}
\end{equation}
Since $\mbs{r}'$ and $\shear$ are perpendicular, the shear resultant can be found to be
\begin{equation}
  \shear =
  \frac{\mbs{d}}{|\mbs{r}'|}
  \times
  \left(\mbs{m}' - \dot{\mbs{\pi}}\right)
  =
  \frac{\mbs{d}}{|\mbs{r}'|}
  \times
  \left(
    \nabla_{\mbs{d}'}\mbs{m}
    - 
    \nabla_{\dot{\mbs{d}}}\mbs{\pi}
  \right)
  \label{eq-proof-kel-3}
\end{equation}
which can be inserted in \eqref{eq-el-k1-1} to obtain Eq.~\eqref{eq-el-k3-1}. Then, projecting
Eq.~\eqref{eq-proof-kel-2} onto the $\mbs{d}$ direction, Eq. \eqref{eq-el-k3-2} follows. It bears
emphasis that the covariant derivative ensures that the differentiation is compatible with
the underlying manifold structure.     
\end{proof}

Eqs.~\eqref{eq-el-k3} are derived solely from variational arguments, including the correct boundary
conditions. One possible simplification would be to ignore the contribution of the rotational
inertia, leading to a model discussed by Meier~\cite{Meier:2014dk}, and obtained directly as a
projection of the Simo-Reissner beam theory. Another possible simplification could be
to ignore the extensibility of the beam. A variational derivation of the corresponding
governing equations would be almost identical to the one presented above. However, in this
case, both the axial and shear strains vanish and thus a constraint must be included
that not only imposes that $\mbs{r}'$ is parallel to $\mbs{e}_3$, as in
Eq.~\eqref{eq-constrain-shear}, but rather that these two quantities are identical.
To impose this restriction, the constraint~\eqref{eq-constrain-shear} must be
replaced by the stronger one
\begin{equation}
  \mbs{0}
  =
  \mbs{\Lambda}^T \mbs{r}' - \mbs{E}_3\ .
  \label{eq-kirchhoff-inext}
\end{equation}
Since these are three constraints that must be satisfied pointwise, the field of
Lagrange multipliers that would be required to impose them would be
functions from $[0,L]$ to $\mathbb{R}^3$. A variational principle for
this class of inextensible, Kirchhoff rods has been presented by Antman \cite{Antman:1995wm}.

\subsection{Linearization}
The governing equation of the Kirchhoff rod, as given by Theorem~\ref{thm-k-g} can be
linearized resulting in the Rayleigh model \cite{Han:1999}, as shown in the following result:

\begin{theorem}
  \label{thm-k-linear}
  Consider a straight rod with constant cross section, aligned in its reference
  configuration with the $\mbs{E}_3$ axis and let its centerline position be written as
  \begin{equation}
    \mbs{r}(s,t) = u(s,t)\,\mbs{E}_1 + v (s,t)\,\mbs{E}_2 + (s+w(s,t))\,\mbs{E}_3\ .
    \label{eq-lin-r}
  \end{equation}
  Assuming that the torsion angle is~$\varphi$, the linearized Euler-Lagrange
  equations from Theorem \ref{thm-k-g} are:
  \begin{subequations}
  \label{eq-lin-el-k3}
  \begin{align}
    \varrho A \ddot{u}-\varrho I_{22} \ddot{u}''+EI_{22} u'''' &= \bar{n}_1\, ,
    \label{eq-lin-ray1}\\
    \varrho A \ddot{v}-\varrho I_{11} \ddot{v}''+EI_{11} v''''&= \bar{n}_2\, ,
    \label{eq-lin-ray2}\\
    \varrho A \ddot{w}+EA{w}''&= \bar{n}_3\, ,
                                \label{eq-lin-ray3}\\
    \varrho I_{33} \ddot{\varphi}+GI_{33}\varphi''&= 0\, .
    \label{eq-lin-ray4}
  \end{align}
\end{subequations}%
Eqs.~\eqref{eq-lin-ray1}-\eqref{eq-lin-ray2} correspond to Rayleigh's beam equations for bending, including
rotational inertia \cite{Han:1999};  Eqs.~\eqref{eq-lin-0-bc-k3-3}-\eqref{eq-lin-0-bc-k3-4} model,
respectively, the axial and torsional behavior. The attendant boundary conditions
at $s=0$ are
\begin{subequations}
  \begin{align}
    u(0) &= 0\, \quad \textrm{and} \quad u'(0) = 0 \ , & \label{eq-lin-0-bc-k3-1}\\
    v(0) &= 0\, \quad \textrm{and} \quad v'(0) = 0 \ , & \label{eq-lin-0-bc-k3-2}\\
    w(0) &= 0\ , & \label{eq-lin-0-bc-k3-3}\\
    \varphi(0) &= 0\ . & \label{eq-lin-0-bc-k3-4} 
  \end{align}
  \label{eq-lin-0-bc-k3}
\end{subequations}
Lastly, the natural boundary conditions at $s=L$ are:
\begin{subequations}
  \begin{align}
    \varrho A\ddot{u}'+EI_{22}u''' &= \tilde{n}_1\, \quad \textrm{and} \quad EI_{22}u'' = 0 \ , & \label{eq-lin-L-bc-k3-1}\\
    \varrho A\ddot{v}'+EI_{11}v''' &= \tilde{n}_2\, \quad \textrm{and} \quad EI_{11}v'' = 0 \ , & \label{eq-lin-L-bc-k3-2}\\
    EAw' &= \tilde{n}_3\ , & \label{eq-lin-L-bc-k3-3}\\
    GI_{33}\varphi' &= 0\ . & \label{eq-lin-L-bc-k3-4} 
  \end{align}
  \label{eq-lin-L-bc-k3}
\end{subequations}
\end{theorem}
\begin{proof}
Let us introduce the smallness parameter $\varsigma$ and redefine the displacement fields $u$,
$v$ and $w$ as $\varsigma \bar{u}$, $\varsigma \bar{v}$ and $\varsigma \bar{w}$,
respectively. Similarly, the torsion angle $\varphi$ is redefined as
$\varsigma\bar{\varphi}$. Such re-scaled fields allow us to reveal the order of each term in the
smallness parameter ($\varsigma^n$ with $n = 0,1,\ldots,N,\ldots,\infty$). The linearization
consists in retaining only those terms up to order $n=1$ and therefore, every term for $n>1$ is then
collected by the term $\Order{\varsigma}{2}$ since no further order distinction is necessary.

Next, we employ the asymptotic expansion to linearize the governing equation and boundary
conditions. For this purpose, the rotation tensor associated to the cross section of the rod can be
expressed as
\begin{equation}
\mbs{\Lambda} = \mbs{I}-
\varsigma \bar{v}'\hat{\ui}+
\varsigma\bar{u}'\hat{\uj}+
\varsigma\bar{\varphi}\hat{\uk}+
\Order{\varsigma}{2} \,. 
\end{equation}
Let us not that $\mbs{\Lambda}$ is not an arbitrary rotation tensor, but one that guarantees that
normal vector to the cross section remains parallel to $\mbs{r}'$ and rotates about that same
direction an angle $\varphi$ with respect to the natural frame.

The following expansions follow directly from the definition of the nonlinear terms:
\begin{subequations}\label{eq-lins}
   \begin{align}
     \bm{d} &= \frac{\bm{r}'}{|\bm{r}'|} = \uk + \varsigma \bar{u}'\ui
              + \varsigma \bar{v}'\uj + \varsigma \bar{w}'\uk +
     \Order{\varsigma}{2}\, ,
     \label{eq-lin-1}\\
     \bm{p} &= \varrho A \varsigma \dot{\bar{u}}\ui + \varrho A \varsigma
     \dot{\bar{v}}\uj + \varrho A \varsigma \dot{\bar{w}}\uk\, ,
     \label{eq-lin-p}\\
     \bm{\pi} &= -\varrho I_{11}\varsigma\dot{\bar{v}}'\ui + \varrho I_{22}\varsigma\dot{\bar{u}}'\uj
     +
     \varrho I_{33}\varsigma\dot{\bar{\varphi}}\uk+\Order{\varsigma}{2}\, ,
     \label{eq-lin-pi}\\
     \frac{\bm{d}}{|\bm{r}'|} \times\nabla_{\dot{\bm{d}}}{\bm{\pi}}
     &=
     -\varrho I_{22}\varsigma\ddot{\bar{u}}'\ui-\varrho I_{11}\varsigma\ddot{\bar{v}}'\uj+\Order{\varsigma}{2}\, ,
     \label{eq-lin-drcovpi}\\
     \axial &= EA \varsigma \bar{w}'\uk+\Order{\varsigma}{2}\, ,
     \label{eq-lin-na} \\
     \bm{m} &= -EI_{11}\varsigma\bar{v}''\ui+EI_{22}\varsigma\bar{u}''\uj+GI_{33}\varsigma\bar{\varphi}'\uk+\Order{\varsigma}{2}\, ,
     \label{eq-lin-m} \\
     \frac{\bm{d}}{|\bm{r}'|}\times  \nabla_{\bm{d}'} {\bm{m}} &=
     -E I_{22}\varsigma\bar{u}'''\ui-E I_{11}\varsigma\bar{v}'''\uj+\Order{\varsigma}{2}\, .
     \label{eq-lin-drcovm}
   \end{align}
\end{subequations}
Replacing these expansions in Eqs.~\eqref{eq-el-k3}-\eqref{eq-L-bc-k3}, removing the terms
of order $\Order{\varsigma}{2}$, and recovering the original fields $u$, $v$, $w$ and $\varphi$, we
finally obtain Eqs.~\eqref{eq-lin-el-k3}-\eqref{eq-lin-L-bc-k3}.
\end{proof}

\section{Interlude: Kirchhoff rods theories with composite rotations}
\label{sec-composite}
The Kirchhoff rod models described in Section~\ref{sec-canonical} are formulated in a configuration
space~\eqref{eq-Q} identical to the one employed in rods with shear. The Kirchhoff constraint is
then variationally included in the action and Lagrange multipliers are added to impose it.

Several works can be found in the literature where the use of Lagrange multipliers is avoided by
defining a new configuration space, one in which the constraint is already accounted for, sparing
the need for Lagrange multipliers in the variational formulation (\eg,\cite{Boyer:2004,Boyer:2011hv,
Greco:2013co, Meier:2014dk, Meier:2017gs}). If $\mbs{d}=\mbs{r}'/|\mbs{r}'|$, these references
define the rotation of each cross section as a composite map
\begin{equation}
  \mbs{\Lambda} = \exp[\psi\, \hat{\mbs{d}}]\, \mbs{R}[\mbs{d}_0,\mbs{d}] .
  \label{eq-composite}
\end{equation}
One rotation $\mbs{R}$ maps $\mbs{d}_0$ to $\mbs{d}$ and is followed by a rotation about
$\mbs{d}$ of angle $\psi$. Such a rotation definition requires only a smooth field $\mbs{r}$ plus an additional
scalar for the second rotation, avoiding a full parameterization of the rotation field and
the use of Lagrange multipliers altogether.

These formulations seem very appealing, since they not only remove the rotation group from the
configuration space, but avoid the use of Lagrange multipliers. They suffer however, from three
drawbacks that might be important. First, this kind of composite parameterizations often has
singularities. A common choice for $\mbs{R}[\mbs{d}_0,\mbs{d}]$ is $\mbs{\chi}[\mbs{d}_0,\mbs{d}]$,
the unique rotation that maps $\mbs{d}_0$ to $\mbs{d}$ without drill defined in
Eq.~\eqref{eq-no-drill}, that is undefined when $\mbs{d} = - \mbs{d}_0$. Such singularities might
not be relevant for problems with small rotations, or incremental solutions. However, a complete,
self-contained theory of Kirchhoff rods cannot be based on a singular parameterization.

The second drawback of composite parameterizations is related to the imposition of boundary
conditions. Taking the derivative with respect to the arc-length of Eq.~\eqref{eq-composite} the
curvature $\mbs{K}$ and torsional strain $\tau$ can be calculated. A careful derivation shows that,
if $\mbs{R}$ is obtained with expression~\eqref{eq-no-drill}, the strain $\mbs{\omega}$ is
of the form
\begin{equation}
  \mbs{\omega} = k_1 \mbs{e}_1 + k_2 \mbs{e}_2 + (\psi' + \phi) \mbs{d}\ ,
  \label{eq-dis-omega}
\end{equation}
with $k_\alpha = (\mbs{d}\times\mbs{d}')\cdot\mbs{e}_\alpha$ for $\alpha=1,2$, and the torsional
part of the strain is not $\psi'$ but has and additional term $\phi$ that comes from the torsion of
the drill-free frame, as discussed in Section~\ref{subs-bishop}.  This new term might result in
difficulties for imposing boundary conditions, since the variationally consistent Dirichlet boundary
conditions for angles involve $\psi$ and a complex functions of $\mbs{r}$ and~$\mbs{r}'$.

One could use an intermediate rotation $\mbs{R}$ whose torsion is zero and depends only on the curve
$\mbs{r}$, that is, Bishop's frame $\mbs{B}$.  Using such a frame would remove the indeterminacy in
the boundary conditions, since the strain $\mbs{\omega}$ would be simply of the form
\begin{equation}
  \mbs{\omega} = k_1 \mbs{e}_1 + k_2 \mbs{e}_2 + \psi' \mbs{d}\ ,
  \label{eq-bishop-omega}
\end{equation}
with $\mbs{e}_1=\mbs{B}\mbs{E}_1, \mbs{e}_2 = \mbs{B}\mbs{E}_2$. While using this frame
would solve many of the difficulties alluded to above, there is no explicit
expression for Bishop's frame depending solely on $\mbs{d}_0$ and $\mbs{d}$.

The third drawback is more subtle and has not been identified, to the authors' knowledge,
previously. When using composite frames, the angular velocity is a complex function.
For a Bishop (natural) frame one might expect that the angular velocity of the cross section
would be of the form
\begin{equation}
  \mbs{w} = w_1 \mbs{e}_1 + w_2 \mbs{e}_2 + \dot{\psi} \mbs{d}\ ,
  \label{eq-bishop-w}
\end{equation}
but this is not the case. The component of the angular velocity in the $\mbs{d}$ direction can not
be calculated in closed form, and a Lagrangian can not be formulated in terms
of $\mbs{r},\psi$ and their derivatives.

To further illustrate this issue, let us compute the angular velocity
of the composite rotation~\eqref{eq-composite}. It has the form
\begin{equation}
  \mbs{w}_{\rotation} = \mbs{d}\times\dot{\mbs{d}}+\left(\dot{\psi}+\mbs{a}\cdot\dot{\mbs{r}}'\right)\mbs{d}\ , 
\end{equation}
where
\begin{equation}
\mbs{a} =-\frac{1}{|\mbs{r}'|}\frac{\mbs{d}_0\times\mbs{d}}{1+\mbs{d}_0\cdot \mbs{d}}\ . 	
\end{equation}
If the angular velocity is to have a simple expression of the form~\eqref{eq-bishop-w}, we could try to
replace the spin rotation $\exp[\psi \hat{\mbs{d}}]$ in Eq.~\eqref{eq-composite} by
$\exp[ (\psi+\xi) \hat{\mbs{d}}]$, so that the new angular velocity would be of the form
\begin{equation}
  \mbs{w}_{\rotation} = \mbs{d}\times\dot{\mbs{d}}+\left(\dot{\psi}+\dot{\xi}+\mbs{a}\cdot\dot{\mbs{r}}'\right)\mbs{d}\ ,
\end{equation}
and the term $\mbs{a}\cdot\dot{\mbs{r}}'$ could be eliminated by imposing the constraint
\begin{equation}
  \dot{\xi}+\mbs{a}\cdot\dot{\mbs{r}}'=0\ .
  \label{eq-rotation-constraint}
\end{equation}
As proven in Section~\ref{subs-bishop}, a restriction of this type is non-integrable. Moreover, choosing
such a correction angle might spoil the simple form of the (spatial) curvature of the frame.

To see this more explicitly, consider the convected strain $\mbs{\Omega}$ and
angular velocity $\mbs{W}$ of the composite frame whose components in the $\mbs{E}_3$
direction are, respectively, $\psi'$ and $\dot{\phi}$. A general result for the
compatibility of strain and angular velocity  \cite{Coleman:1993hm} states that
\begin{equation}
  \dot{\mbs{\Omega}} - \mbs{W}' = \mbs{\Omega}\times \mbs{W}\ .
  \label{eq-ti-compatibility}
\end{equation}
From this relation, a straightforward manipulation gives
\begin{equation}
  \dot{\psi}' - \dot{\phi}' = \mbs{d}\cdot \mbs{d}'\times \dot{\mbs{d}}\ ,
  \label{eq-ti-comp2}
\end{equation}
proving that, in general, $\psi\ne\phi $.



\section{Variational principles for transversely isotropic Kirchhoff rods}
\label{sec-ti}
In this section we study rods that have transversely isotropic cross
sections, in the sense that the stored energy function and the
in-plane inertia are of the form
\begin{equation}
  \tilde{U}(\epsilon, \mbs{K}, \tau; s)=
  \bar{U}(\epsilon, |\mbs{K}|, \tau; s)
  \ ,
  \qquad
  \mbs{I}_\perp = I_\perp (\mbs{E}_1\otimes \mbs{E}_1 + \mbs{E}_2\otimes \mbs{E}_2)\ .
  \label{eq-ti-energy}
\end{equation}
We note that $|\mbs{K}| = |\mbs{\kappa}| = |\mbs{e}_3'|$ is a scalar bending strain
that is frame invariant. This class of Kirchhoff rods are interesting for two
reasons: first, they are very common in applications and second, they admit certain simplifications
in their formulation as compared with the models presented in Section~\ref{sec-canonical}. In
particular, the quasistatic formulation admits a greatly simplified variational principle. The
dynamic case, however, demands a careful consideration.

\subsection{Quasistatic problems}
The simple form of the stored energy function for transversely isotropic rods allows to bypass
completely the use of rotations in the formulation, although a new unknown field, the twist $\psi$,
is added to account for the torsional deformation. Geometrically speaking, the isotropy of the
bending response avoids the need to \emph{pull-back} the bending strain to the reference
configuration in order to calculate the bending moment, since it is known \emph{a priori} that
this moment is parallel to the bending strain. Since the rotation tensor is needed for
this transformation and for the computation of the torsional strain, only the latter
remains necessary. In the following derivation, however, we avoid the use of the rotation altogether
by introducing an additional configuration term.

The configuration space of this model, for the rod clamped at $s=0$, is the manifold
\begin{equation}
  Q :=
  \left\{
    (\mbs{r},\psi):[0,L]\to \mathbb{R}^3\times \reals,\
    \mbs{r}(0)=\mbs{0},\: \mbs{r}'(0) = \mbs{E}_3,\: \psi(0) = 0
  \right\}\ ,
  \label{eq-ti-q}
\end{equation}
and the potential energy is the functional
\begin{equation}
  V = \int_0^L \bar{U}(\epsilon, |\mbs{K}|, \tau; s) \,\mathrm{d} s
  - \int_0^L ( \bar{\mbs{n}}\cdot \mbs{r} + \bartorsion\psi )
  \, \mathrm{d} s -\tilde{\mbs{n}}\cdot\mbs{r}(L)
  - \tildetorsion \psi(L)\ ,
  \label{eq-ti-v}
\end{equation}
with $\bar{\mbs{n}},\bartorsion$ being, respectively,
known fields of forces and tangent moment per unit length,
and $\tilde{\mbs{n}},\tildetorsion$ a known point force
and a known tangent moment at the end $s=L$.
In this case, the strain measures have the simple form
\begin{equation}
  \epsilon := \mbs{r}'\cdot\mbs{d} - 1\ ,
  \qquad
  |\mbs{K}| := |\mbs{d}'| \ ,
  \qquad
  \tau := \psi'
  \label{eq-ti-strains}
\end{equation}
where, as before, $\mbs{d} := \mbs{r}'/|\mbs{r}'|$. The governing equations
of this model are obtained as the stationarity conditions of~$V$.

\subsubsection{Strain variations}
The strains defined in Eq.~\eqref{eq-ti-strains} are based on the
strain measures of the general rod model (cf. Section~\ref{subs-strains}). However,
since the rotation is no longer an independent field, the variations
also need to be redefined. To calculate the strain variations we introduce,
as in Section~\ref{subs-k-variations}, a curve of perturbed configurations in~$Q$
\begin{equation}
  \left(\mbs{r}_\iota(s), \psi_\iota(s)\right)
  =
  \left(
    \mbs{r}(s) + \iota\, \delta\mbs{r}(s) ,
    \psi(s) + \iota\, \delta\psi(s)
  \right)\ ,
  \label{eq-ti-one-parameter}
\end{equation}
where $\iota\in\reals$, $\delta \mbs{r}:[0,L]\to\reals^3$ and $\delta\psi :[0,L]\to\reals$ are arbitrary fields
that satisfy
\begin{equation}
  \delta \mbs{r}(0) = \mbs{0}\ ,
  \qquad
  \delta \psi(0) = 0 \ .
  \label{eq-ti-variations}
\end{equation}
The advantage of the new configuration space is apparent already, since the definition
of the perturbed configurations is additive, avoiding the use of the exponential map.

The variations of the strain measures are presented in the following result.

\begin{theorem}
  \label{thm-ti-variations}
  The variations of the strains in the transversally isotropic Kirchhoff rod
  are:
  \begin{equation}
  \begin{split}
    \delta \epsilon &= \mbs{d}\cdot \delta \mbs{r}' ,\\
    \delta |\mbs{K}| &= \mbs{d}\times \frac{\mbs{d}'}{|\mbs{d}'|}\cdot \delta \mbs{\beta}' , \\
    \delta \tau &= \delta \psi' . \\
  \end{split}
   \label{eq-thm-ti-vars}
 \end{equation}
 with
 \begin{equation}
   \delta \mbs{\beta} := \frac{1}{|\mbs{r}'|} \mbs{d}\times \nabla_{\mbs{d}'} \delta \mbs{r}\ .
  \label{eq-ti-beta}
\end{equation}
\end{theorem}
\begin{proof} The variations of $\epsilon, |\mbs{K}|$ and $\tau$ are obtained by a systematic
application of the chain rule and the definition of the perturbed
configurations~\eqref{eq-ti-one-parameter}. For convenience, we write
\begin{equation}
  \delta \mbs{d} =
  \left.\pd{}{\iota}\right|_{\iota=0} \frac{ \mbs{r}_\iota'}{| \mbs{r}_\iota' |}
  =\frac{1}{|\mbs{r}'|} \left( \delta \mbs{r}' - (\mbs{d}\cdot \delta \mbs{r}') \mbs{d} \right)
  = \left(\frac{1}{|\mbs{r}'|} \mbs{d}\times \delta \mbs{r}'\right) \times \mbs{d}\
  = \left(\frac{1}{|\mbs{r}'|} \mbs{d}\times \nabla_{\mbs{d}'} \delta \mbs{r} \right)\times \mbs{d}
  \:,
  \label{eq-ti-beta-def}
\end{equation}
and we define the last parenthesis to be equal to $\delta \mbs{\beta}$, an arbitrary
vector field on the rod, orthogonal to $\mbs{d}$. The variation of $\epsilon$ then follows
trivially.

To calculate the variation of the bending strain~$|\mbs{K}|$, it is convenient to write a
one-parameter curve of directors in the form
\begin{equation}
  \mbs{d}_\iota = \exp[\iota\, \widehat{\delta\mbs{\beta}}] \mbs{d}\ ,
  \label{eq-ti-d-param}
\end{equation}
whose derivative with respect to the arc-length is
\begin{equation}
  \mbs{d}_\iota'
  =
  \skew\left[ \dexp[\iota\,\widehat{\delta \mbs{\beta}}] \iota \delta \mbs{\beta}' \right]
  \exp[\iota\, \widehat{\delta \mbs{\beta}}] \mbs{d} + \exp[\iota\, \widehat{\delta \mbs{\beta}}]
  \mbs{d}'\ .
  \label{eq-ti-proof-1}
\end{equation}
Thus, the variation of $\delta \mbs{d}'$ is
\begin{equation}
  \delta \mbs{d}'
  =
  \left.\pd{}{\iota}\right|_{\iota=0} \mbs{d}_\iota' = \delta \mbs{\beta}'\times \mbs{d}
  +
  \delta \mbs{\beta}\times \mbs{d}' ,
  \label{eq-ti-proof-2}
\end{equation}
which, together with
\begin{equation}
  \delta |\mbs{K}| = \delta |\mbs{d}'| = \frac{1}{|\mbs{d}'|} \mbs{d}'\cdot \delta \mbs{d}'\ ,
  \label{eq-ti-proof-3}
\end{equation}
yields the variation of the second strain. The variation of the torsional strain is trivial.
\end{proof}

\subsubsection{Quasistatic equilibrium equations}
The equilibrium equations of the rod are obtained from the
Euler-Lagrange equations of the potential energy.

\begin{theorem}
  The equilibrium equations for a transversely isotropic Kirchhoff rod
  are:
  \begin{equation}
    \begin{split}
      \left(
        \axial +  \frac{\mbs{d}}{|\mbs{r}'|}\times \nabla_{\mbs{d}'}  \bending
      \right)'
         + \bar{\mbs{n}} &= \mbs{0}\ ,\\
       \storsion' + \bartorsion &= 0\ ,
  \end{split}
   \label{eq-ti-equilibrium}
 \end{equation}
 where the stress resultants are defined as
 \begin{equation}
   \axial = \pd{\bar{U}}{\epsilon} \mbs{d}
   \ ,\qquad
   \bending = \pd{\bar{U}}{|\mbs{K}|} \mbs{d}\times \frac{\mbs{d}'}{|\mbs{d}'|}
   \ ,\qquad
   \storsion = \pd{\bar{U}}{\tau}\ ,
  \label{eq-ti-theo-1}
\end{equation}
and the natural boundary conditions at $s=L$ are
\begin{equation}
  \axial  +  \frac{\mbs{d}}{|\mbs{r}' |}\times \nabla_{\mbs{d}'}  \bending 
  = \tilde{\mbs{n}}
   \ ,\qquad
   \bending\times \frac{\mbs{d}}{|\mbs{r}'|} = \mbs{0}
   \, \qquad
    \storsion = \tildetorsion\ .
   \label{eq-ti-theo-2}
\end{equation}
\end{theorem}

\begin{proof}
The theorem follows from a systematic calculation of $\delta V$, the
variation of the potential energy, which is developed next:
\begin{equation}
  \begin{split}
    \delta V
    &=
    \int_0^L
    \left[
      \pd{\bar{U}}{\epsilon}\;\delta\epsilon +
      \pd{\bar{U}}{|\mbs{K}|}\;\delta |\mbs{K}|  +
      \pd{\bar{U}}{\tau}\;\delta\tau  -
      \bar{\mbs{n}}\cdot \delta \mbs{r} -
      \bartorsion \delta\psi
    \right]
    \,\mathrm{d} s
    -
    \tilde{\mbs{n}}\cdot \delta\mbs{r}(L)
    -
    \tildetorsion \delta\psi(L)\\
    &=
    \int_0^L
    \left[
      \pd{\bar{U}}{\epsilon}\; \mbs{d}\cdot \delta \mbs{r}' +
      \pd{\bar{U}}{|\mbs{K}|}\;\mbs{d}\times \frac{\mbs{d}'}{|\mbs{d}'|}\cdot \delta \mbs{\beta}'\ +
      \pd{\bar{U}}{\tau}\;\delta\psi'  -
      \bar{\mbs{n}} \cdot \delta \mbs{r}-
      \bartorsion \delta\psi
    \right]
    \,\mathrm{d} s
    -
    \tilde{\mbs{n}}\cdot \delta\mbs{r}(L)
    -
    \tildetorsion \delta\psi(L) .\\
  \end{split}
   \label{eq-ti-proof-4}
\end{equation}
Defining the axial force, bending, and torsion moment, respectively, as
in Eq.~\eqref{eq-ti-theo-1} and integrating by parts, it follows that
\begin{equation}
  \begin{split}
    \delta V =&
    \int_0^L
    \left[
      - (\axial' + \bar{\mbs{n}})\cdot \delta \mbs{r}
      - \nabla_{\mbs{d}'}\bending\cdot \delta \mbs{\beta}
      - (\storsion'+ \bartorsion) \delta\psi
    \right] \,\mathrm{d} s
    \\
    &+
    (\axial(L) - \tilde{\mbs{n}})\cdot\delta \mbs{r}(L)
    +
    \bending(L)\times \frac{\mbs{d}(L)}{|\mbs{r}'(L)|}\cdot \delta \mbs{r}'(L)
        +
    (\storsion(L) - \tildetorsion) \delta\psi(L)
    \:.
  \end{split}
   \label{eq-ti-proof-6}
\end{equation}
Finally, expanding $\delta \mbs{\beta}$ as in Eq.~\eqref{eq-ti-beta}, and
integrating by parts a second time we arrive at the final expression for
the variation:
\begin{equation}
  \begin{split}
    \delta V
    &=
    \int_0^L
    \left[
      -
      \left(
        (\axial + \frac{\mbs{d}}{|\mbs{r}'|}\times \nabla_{\mbs{d}'}\bending)' +
        \bar{\mbs{n}}
      \right) \cdot \delta \mbs{r}
      -
      (\storsion' + \bartorsion) \delta\psi
    \right] \,\mathrm{d} s
    \\
    &\qquad
    +
    (\axial(L) + \frac{\mbs{d}(L)}{|\mbs{r}'(L)|}\times \nabla_{\mbs{d}'}\bending (L) - \tilde{\mbs{n}})
    \cdot\delta \mbs{r}(L)
    \\
    &\qquad
    +
    \bending(L)\times \frac{\mbs{d}(L)}{|\mbs{r}'(L)|} \cdot \delta \mbs{r}'(L)
    \\
    &\qquad
    +
    (\storsion(L) - \tildetorsion) \delta\psi(L)\ .
  \end{split}
   \label{eq-ti-proof-7}
\end{equation}
Using the condition $\delta V=0$ and the arbitrariness of the variations, the
theorem is proved.
\end{proof}

\subsection{Difficulties with dynamic problems}
The variational principle of transversely isotropic Kirchhoff rods can be recovered from the general
case when the rotation field is chosen to be the composition of Bishop's frame and a
rotation of angle $\psi$ around the tangent vector $\mbs{d}$. As a result of the isotropy, Bishop's
frame does not play any role and its calculation can be avoided, leaving a model with only four
unknown fields, namely, the three components of $\mbs{r}$ and $\psi $.

The simplicity of the four-field formulation makes it very appealing for its use in numerical
computations and it has been exploited in many works. Its extension to dynamical problems has also
been considered, but not always in a fully correct fashion.

In a transversely isotropic rod the strain measure $\mbs{\omega}$
can be expressed as
\begin{equation}
  \mbs{\omega}
  =
  \mbs{d}\times \mbs{d}' + \psi'\, \mbs{d}
  \label{eq-ti-omega}
\end{equation}
and $\mbs{\Omega}= \mbs{\Lambda}^T \mbs{\omega}$, when $\mbs{\Lambda}$
is the composite rotation based on Bishop's frame. In view
of the parallelism between the strain $\mbs{\omega}$ and the
spatial angular velocity $\mbs{w}$, it is tempting to assume that
the latter must be of the form
\begin{equation}
  \mbs{w}
  =
  \mbs{d}\times \dot{\mbs{d}}  + \dot{\phi}\, \mbs{d}\ ,
  \label{eq-ti-w}
\end{equation}
where the spin velocity satisfies $\dot{\phi} = \dot{\psi}$. 
This formula for the angular velocity has been used in the past \cite{Greco:2013co}, but it is
not true, in general, as discussed in Section~\ref{sec-composite}.

\subsection{Variational formulation of dynamic problems}
\label{subs-ti-dynamic}
It remains to find a variational formulation for dynamic problems of
transversely isotropic Kirchhoff rods,
one that take advantage of the four-field formulation. As discussed before, this
is not straightforward because there is no simple way to express the kinetic
energy in terms of $\mbs{r}$ and $\psi $. To formulate the full kinetic energy it
is thus necessary to introduce the rotation field again in the variational principle.

One might keep the advantages of the four-field formulation by neglecting
the kinetic energy associated to the spin around $\mbs{d}=\mbs{r}'/|\mbs{r}'|$, \ie, assuming its
density to be of the form
\begin{equation}
  k = \frac{1}{2} A_\rho |\dot{\mbs{r}}|^2 + \frac{1}{2} \mbs{w}_\perp\cdot \mbs{i}_\perp
  \mbs{w}_\perp\ ,
  \label{eq-ti-kinetic}
\end{equation}
where $\mbs{w}_\perp = \mbs{d}\times \dot{\mbs{d}}$. In view of the isotropy
property of the beam, the latter can be written simply as
\begin{equation}
  k = \frac{1}{2} A_\rho |\dot{\mbs{r}}|^2 + \frac{1}{2} {i}_\perp |\dot{\mbs{d}}|^2\ .
  \label{eq-ti-kinetic2}
\end{equation}

This approximation on the form of the kinetic energy precludes the use of the model for problems in
which a rod spins around it centerline although it might be an admissible assumption for many
practical problems (cable deployment, DNA modeling, etc.)
and has often been neglected tacitly (\eg, \cite{Valverde:2006cj}).
In any case, since the field $\psi$ will appear in the governing equations of the model,
but not its rate nor acceleration, the latter will admit torsional waves of infinite speed. In
order to recover the structure of an hyperbolic initial boundary-value problem one might choose to
\emph{regularize} the kinetic energy~\eqref{eq-ti-kinetic2} by defining it to be
\begin{equation}
  k = \frac{1}{2} A_\rho |\dot{\mbs{r}}|^2 + \frac{1}{2} {i}_\perp |\dot{\mbs{d}}|^2
  +
  \frac{1}{2} i_\parallel \dot{\psi}^2\ ,
  \label{eq-ti-kinetic-reg}
\end{equation}
recognizing from the outset that this is not the exact expression for the kinetic energy but
merely a convenient approximation.  We present next the variational formulation and governing
equations of the transversely isotropic Kirchhoff rod with the approximated kinetic
energy~\eqref{eq-ti-kinetic-reg}. As in the general case, we use
Hamilton's principle of stationary action to derive the governing equations for the rod under
consideration.

In contrast with the model presented in Section~\ref{sec-canonical},
the variational principle for the model under consideration does not include any constraint.
The action functional is of the form
\begin{equation}
  S = \int_0^T (T - V) \,\mathrm{d} t\ ,
  \label{eq-ti-S}
\end{equation}
where the kinetic energy~$T$ is the integral over the rod
of the density~$k$ as given by
Eq.~\eqref{eq-ti-kinetic-reg}
and the potential energy~$V$ is defined in Eq.~\eqref{eq-ti-v}.
The Euler-Lagrange equations of the action~\eqref{eq-ti-S} are
obtained using the same type of calculations as in previous
models. We present without proof the main result:
\begin{theorem}
  The dynamic equilibrium equations of a transversely isotropic
  Kirchhoff rod are:
  \begin{equation}
    \begin{split}
      \left(
        \axial +  \frac{\mbs{d}}{|\mbs{r}'|}\times \nabla_{\mbs{d}'}  \bending
      \right)'
      + \bar{\mbs{n}} &=
      \dot{\mbs{p}}
      + \left( \frac{\mbs{d}}{|\mbs{r}'|}\times \nabla_{\dot{\mbs{d}}} \mbs{\pi}_\perp \right)'
      \ , \\
      \storsion' + \bartorsion &= \dot{\pi}_\parallel\ ,
    \end{split}
    \label{eq-ti-dyn-1}
  \end{equation}
  where the stress resultants are defined in Eq.~\eqref{eq-ti-theo-1}, and the natural boundary
  conditions at $s=L$ are
  \begin{equation}
    \axial +  \frac{\mbs{d}}{|\mbs{r}'|}\times
    \left( \nabla_{\mbs{d}'}  \bending - \nabla_{\dot{\mbs{d}}} \mbs{\pi}_\perp \right)
    = \tilde{\mbs{n}}
   \ ,\qquad
   \bending\times \frac{\mbs{d}}{|\mbs{r}'|} = \mbs{0}
   \, \qquad
    \storsion = \tildetorsion\ .
    \label{eq-ti-dyn-2}
  \end{equation}
\end{theorem}

We conclude this section reflecting upon whether the use of dynamic formulations for transversely
isotropic Kirchhoff rods is a sensible choice or not. As discussed at length, this type of models
always entails approximations on the kinetic energy that can only be justified by
arguing that the torsional effects play a small role in the model. In these cases, however, one wonders if
it would not be better to employ, from the outset, a rod model that completely eliminates the torsion
from the equations \cite{RoUrreCy:2014}.
\section{A joint variational principle for general and transversely isotropic Kirchhoff rods}
\label{sec-mixed}
Sections~\ref{sec-canonical} and \ref{sec-ti} discuss, respectively, variational principles for
general and transversely isotropic Kirchhoff rods. The configuration space defined for the first
type of rods is too rich, and Lagrange multipliers have to be used in the variational principle to
constrain the rod's kinematics. This is in contrast with the configuration space of the second kind,
which fits exactly the kinematics of the rod and makes unnecessary to employ constraints in the
corresponding variational principle. On the downside, this latter formulation is only valid for
transversely isotropic rods.

In this section we propose a third variational principle that is valid for rods of
general cross sections, but simplifies when the model is transversely isotropic. As with the
principles of Section~\ref{sec-ti}, the new functional is strictly exact only for
quasistatic problems. Its extension to dynamic problems requires that the part of the
kinetic energy associated with the rotation of the cross section about its director be
neglected. 

The idea of the new developments is to use, as in previous sections, Hamilton's principle of
stationary action to obtain the equations of motion of Kirchhoff rods, but in a way that the
solution of certain fields can be decoupled from the rest of unknowns \emph{when the rod is
transversely isotropic}.  In the general case, the variational principle is completely equivalent to
the one studied in Section~\ref{sec-canonical} if the problem is quasistatic or the contribution
from the inertia mentioned above is ignored.

The new principle is defined in terms of the unknown fields $\mbs{r},\mbs{\Lambda},\psi: [0,L]\to
\reals^3\times SO(3)\times \reals$ and the Lagrange multipliers $\mbs{\eta},\mu:[0,L]\to
\reals^2\times\reals$. The unknown fields refer, as in previous sections, to the position of the
curve of centroids, the section orientation, and the torsional angle. Based on these, the potential
energy of the rod has the usual form
\begin{equation}
  V
  = \int_0^{L} \tilde{U}(\epsilon, \mbs{K}, \tau; s) \,\mathrm{d} s
  - \int_0^{L} (\bar{\mbs{n}}\cdot r +\bartorsion\, \psi )   \,\mathrm{d} s
  - \tilde{\mbs{n}}\cdot \mbs{r}(L)
  - \tildetorsion \psi(L)
  \ ,
  \label{eq-mix-v}
\end{equation}
but now the strain measures are defined as
\begin{equation}
  \epsilon := \mbs{r}'\cdot \mbs{d} - 1
  \ ,\qquad
  \mbs{K} := \mbs{\Lambda}^T \mbs{\kappa}
  \ ,\qquad
  \tau := \psi'\ ,
  \label{eq-mix-strains}
\end{equation}
with
\begin{equation}
  \mbs{d} := \frac{\mbs{r}'}{|\mbs{r}'|}
  \ ,
  \qquad
  \mbs{\kappa} := \mbs{d}\times \mbs{d}'\ .
  \label{eq-mix-strains2}
\end{equation}
We note that the strain measures are completely identical to the ones defined in
Eqs.~(\ref{eq-epsilon})-(\ref{eq-torsion-strain}) and Eq.~(\ref{eq-ti-strains}) when the Kirchhoff
constraint is verified, but the arguments are different. In contrast with the model of
Section~\ref{sec-ti}, the full bending strain enters now the stored energy function, but the role
played by the rotation $\mbs{\Lambda}$ in the definition of $\epsilon$ and $\mbs{K}$ is different
than in the general model of Section~\ref{sec-canonical}.

Next, the kinetic energy is defined to be
\begin{equation}
  T =
  \int_0^{L}
  \left(
    \frac{1}{2} |\dot{\mbs{r}}|^2
    +
    \frac{1}{2} \mbs{w}_\perp \cdot \mbs{i}_\perp \mbs{w}_\perp
  \right)
  \,\mathrm{d} s \ ,
  \label{eq-mix-kinetic}
\end{equation}
with $\mbs{w}_\perp = \mbs{d}\times\dot{\mbs{d}}$.
We recall that the rotational part of this energy can be written with convected objects
using the relation:
\begin{equation}
  \frac{1}{2} \mbs{w}_\perp \cdot \mbs{i}_\perp \mbs{w}_\perp
  =
  \frac{1}{2} \mbs{W}_\perp \cdot \mbs{I}_\perp \mbs{W}_\perp\ .
  \label{eq-mix-kinetic-alternative}
\end{equation}
Finally, to impose Kirchhoff's constraint and that $\psi$ be the torsional angle,
we define the action of the model as the constrained functional:
\begin{equation}
  S := \int_0^T
  \left[
    T - V
    -
    \int_0^{L}
    \left(
      \mu\cdot( \psi'- \mbs{\Omega}\cdot \mbs{E}_3)
      +
      \mbs{\eta}\cdot (\mbs{\Lambda}^T \mbs{r}'\times \mbs{E}_3)
    \right)
    \,\mathrm{d} s
  \right]
  \,\mathrm{d} t\ .
  \label{eq-mixed-action}
\end{equation}
Barring the simplifications in the definition of the kinetic energy,
the actions~\eqref{eq-mixed-action} and~(\ref{eq-kirchhoff-action})
are clearly equivalent, the only difference being the addition
of a new variable $\psi$ defined by the relation $\psi'= \mbs{\Omega}\cdot \mbs{E}_3$.

\subsection{Strain and velocity variations}
The definitions of the strains~(\ref{eq-mix-strains}) and the kinetic energy density in
Eq.~(\ref{eq-mix-kinetic}) are different from the ones of Sections~\ref{sec-canonical}
and~\ref{sec-ti}, although identical when the Kirchhoff constraint is verified. Their variations
will be used in Section~\ref{subs-mix-eq} and need to be calculated anew.

To this end, let us define the configuration space for this principle considering,
for simplicity, once again the case of a rod clamped at the end $s=0$:
\begin{equation}
  Q :=
  \left\{
    (\mbs{r},\mbs{\Lambda},\psi, \mbs{\eta}, \mu): [0,L]\to
    \reals^3 \times SO(3) \times \reals \times \reals^2 \times \reals\ ,
    \
    \mbs{r}(0) = \bar{\mbs{r}},\ 
    \mbs{\Lambda}(0) = \bar{\mbs{\Lambda}},\ 
    \psi(0) = 0
  \right\}\ .
  \label{eq-mix-q}
\end{equation}
Proceeding as in Sections~\ref{sec-canonical} and \ref{sec-ti},
we define the one-parameter curve of configurations
\begin{equation}
  (\mbs{r}_\iota, \mbs{\Lambda}_\iota, \psi_\iota, \mbs{\eta}_\iota, \mu_\iota)
  =
  \left(
    \mbs{r} + \iota\delta\mbs{r},
    \exp[\iota \hat{\delta\mbs{\theta}}]\mbs{\Lambda},
    \psi + \iota\delta\psi ,
    \mbs{\eta} + \iota \delta\mbs{\eta},
    \mu + \iota \delta\mu
  \right)
  \label{eq-mix-perturbed}
\end{equation}
where $\iota\in\reals$, and $\delta \mbs{r},\delta \mbs{\theta}, \delta \psi,
\delta \mbs{\eta}, \delta \mu$ being arbitrary variations with the properties
\begin{equation}
  \delta \mbs{r}(0) = \mbs{0}\ ,\quad
  \delta \mbs{\theta}(0) = \mbs{0},\quad
  \delta \psi(0) = 0\ ,
  \label{eq-mix-var-bc}
\end{equation}
vanishing at $t=0$. Then, a systematic application of the concept of linearization yields
the following result, which we present without proof since it is very
similar to Theorem~\ref{thm-ti-variations}.
\begin{theorem}
  The variations of the strains~(\ref{eq-mix-strains}) are:
  \begin{equation}
   \begin{split}
     \delta\varepsilon &= \mbs{d}\cdot \delta \mbs{r}'   \ , \\
     \delta \mbs{K} &= \mbs{\Lambda}^T\left( \mbs{\kappa}\times \delta \mbs{\beta} +
       \delta \mbs{\theta} \times \mbs{\kappa} + \nabla_{\mbs{d}'} \delta \mbs{\beta}
       \right) \ , \\
     \delta \tau &= \delta\psi'\ ,\\
   \end{split}
   \label{eq-mix-strain-vars}
 \end{equation}
 with $\delta \mbs{\beta}$ defined as in Eq.~(\ref{eq-ti-beta}). Similarly,
 the variation of the convected angular velocity is given by
\begin{equation}
  \delta \mbs{W}_\perp
  =
  \mbs{\Lambda}^T\left( \mbs{w}_\perp\times \delta \mbs{\beta} +
       \delta \mbs{\theta} \times \mbs{w}_\perp + \nabla_{\dot{\mbs{d}}} \delta \mbs{\beta}
     \right) .
     \label{eq-mix-vel-vars}
\end{equation}
\end{theorem}

\subsection{Governing equations}
\label{subs-mix-eq}

The stationarity conditions of the action~(\ref{eq-mixed-action}) are completely
equivalent to Eqs.~(\ref{eq-el-k3}), as the following result shows:

\begin{theorem}
  The Euler-Lagrange equations corresponding to the action~(\ref{eq-mixed-action})
  are:
\begin{subequations}\label{eq-mix-el}
\begin{align}
   \left(
   \axial +
     \frac{\mbs{d}}{|\mbs{r}'|}\times
     \nabla_{\mbs{d}'} \bending
   \right)'
   +\bar{\mbs{n}}
   &=
   \dot{\mbs{p}}
   +
   \left(
     \frac{\mbs{d}}{|\mbs{r}'|}\times
       \nabla_{\dot{\mbs{d}}} \mbs{\pi}_\perp
   \right)'
     \ , \label{eq-mix-el1}
  \\
 (\bending + \storsion \mbs{d})' \cdot\mbs{d}  + \bartorsion
   &=
   \dot{\mbs{\pi}}_\perp \cdot \mbs{d}
    \ ,
     \label{eq-mix-el2}
\end{align}
\end{subequations}
where the stress resultants are defined in Eq.~\eqref{eq-ti-theo-1}.
The attendant natural boundary conditions at $s=L$ are:
\begin{equation}
  \begin{split}
    \axial +
     \frac{\mbs{d}}{|\mbs{r}'|}\times
     \left(
       \nabla_{\mbs{d}'}\bending -
       \nabla_{\dot{\mbs{d}}} {\mbs{\pi}_\perp}
     \right)
     &= \tilde{\mbs{n}}\ ,
     \\
     \frac{\mbs{d}}{|\mbs{r}'|}\times \bending &= 0\ ,
     \\
     \storsion &= \tildetorsion\ .
    \end{split}
    \label{eq-mix-el3}
  \end{equation}
  The natural boundary conditions at $s=L$ follow easily from
  the last integral in Eq.~\eqref{eq-mix-el-proof2} and the
  arbitrariness of the variations at the free end of the beam.
\end{theorem}

\begin{proof}
  The variation of the action gives:
  \begin{equation}
   \begin{split}
     \delta S &=
     \int_0^T
     \left[
       \int_0^L
       \left(
         A_\rho \dot{\mbs{r}}\cdot\delta \dot{\mbs{r}}
         +
         \mbs{I}_\perp \mbs{W}_\perp \cdot \delta \mbs{W}_\perp
       \right)
       \,\mathrm{d} s
     \right.
     \\
     &\quad
     -
     \int_0^L
     \left(
       \pd{U}{\epsilon} \delta\epsilon
       +
       \pd{U}{\mbs{K}}\cdot \delta \mbs{K}
       +
       \pd{U}{\tau} \delta\tau
       -
       \bar{\mbs{n}}\cdot\delta \mbs{r}
       -
       \bartorsion \delta \psi
     \right)
     \,\mathrm{d} s
     + \tilde{\mbs{n}}\cdot \delta \mbs{r}(L)     
     + \tildetorsion \, \delta \psi(L)     
     \\
     &\quad
     -
     \int_0^L
     \left(
       \delta\mu(\psi'-\mbs{\Omega}\cdot \mbs{E}_3)
       +
       \mu(\delta \psi' - \delta \mbs{\Omega}\cdot \mbs{E}_3)
     \right)
     \,\mathrm{d} s \\
     &\quad
     \left.
       - \int_0^L
       \left(
         \delta \mbs{\eta}\cdot
         (\mbs{\Lambda}^T \mbs{r}'\times \mbs{E}_3)
         +
         \mbs{\eta}\cdot( \delta \mbs{\Lambda}^T \mbs{r}'\times \mbs{E}_3
         + \mbs{\Lambda}^T \delta \mbs{r}'\times \mbs{E}_3)
       \right)
       \,\mathrm{d} s
     \right]
     \,\mathrm{d} t\ .
   \end{split}
   \label{eq-mix-el-proof1}
 \end{equation}
 Replacing the formulas~(\ref{eq-mix-strain-vars})-(\ref{eq-mix-vel-vars})
 for the variations,
 using the property that the variations vanish at $t=0$,
 integrating by parts in time the first integral
 and in arc-length the remaning ones, we get, after some straightforward manipulations:
 \begin{equation}
   \begin{split}
     \delta S &=
     \int_0^T
     \int_0^L
     \left[
       - A_\rho \ddot{\mbs{r}}
       +
       \pd{}{s}
       \left(
         \axial
         +
         \frac{\mbs{d}}{|\mbs{r}'|} \times \bending'
         -
         \frac{\mbs{d}}{|\mbs{r}'|} \times \pd{}{t}(\mbs{i}_\perp \mbs{w}_\perp)
       \right)
       + \bar{\mbs{n}}
       \right]
       \cdot
       \delta \mbs{r}
       \,\mathrm{d} s
       \,\mathrm{d} t
     \\
     &\quad
     +
     \int_0^T
     \int_0^L
     \left[
       \mbs{\kappa}\times \bending + \mbs{r}'\times \shear + (\mu\,\mbs{e}_3)'-
       \mbs{w}_\perp\times(\mbs{i}_\perp \mbs{w}_\perp)
     \right]
     \cdot \delta \mbs{\theta}
     \,\mathrm{d} s
     \,\mathrm{d} t
     \\
     &\quad
     +
     \int_0^T
     \int_0^L
     \left[
       \left(
         \storsion' - \mu' + \bartorsion
       \right)
       \delta \psi
       +
       \left(
         \psi' - \mbs{\Omega}\cdot \mbs{E}_3
       \right) \cdot \delta \mbs{\mu}
       +
       \left( \mbs{\Lambda}^T \mbs{r}' \right)\times \mbs{E}_3 \cdot \delta \mbs{\eta}
     \right]
     \,\mathrm{d} s
     \,\mathrm{d} t
     \\
     &\quad
     +
     \int_0^T
     \left[
       \left[\left( 
         \frac{\mbs{d}}{|\mbs{r}'|} \times \pd{}{t}(\mbs{i}_\perp \mbs{w}_\perp)
         -
         \frac{\mbs{d}}{|\mbs{r}'|} \times \mbs{m}_\perp'
         -
         \mbs{n}_\parallel
       \right)\cdot\delta\mbs{r} \right]_0^L
       + \tilde{\mbs{n}}\cdot \delta \mbs{r}(L)
       + \tildetorsion\, \delta \psi(L)
     \right.\\
     &\quad\qquad\left.
       +
       \left[
         \frac{\mbs{d}}{|\mbs{r}'|}\times \mbs{m}_\perp \cdot \delta \mbs{r}'
       \right]_0^L
       -
       \left[
         \mu \mbs{e}_3 \cdot \delta \mbs{\theta}
       \right]_0^L
       +
       \left[
         (\mu-\storsion)\; \delta \psi
       \right]_0^L
     \right]
     \,\mathrm{d} t\ .
   \end{split}
   \label{eq-mix-el-proof2}
 \end{equation}
 Next, we show that the condition $\delta S=0$ gives the same
 differential equations of the general rod model described in Section~\ref{sec-canonical}.
 First, we note that the Lagrange multipliers impose strongly the conditions
 \begin{equation}
   \mu' = \storsion'  + \bartorsion
   \ ,\qquad
   \psi' = \mbs{\Omega}\cdot \mbs{E}_3
   \ ,\qquad
   (\mbs{\Lambda}^T \mbs{r}')\times \mbs{E}_3 = 0\ .
   \label{eq-mix-el-cons}
 \end{equation}
 The first of these conditions provides an interpretation for the Lagrange
 multiplier $\mu$. The second one identifies $\psi$ as the torsion angle. The third one is equivalent to
 the relation $\mbs{d}=\mbs{e}_3$. We will use these relations to simplify the notation
 in the expressions that follow.

 Since the variations $\delta \mbs{r}$ are arbitrary, the first integral in
 Eq.~(\ref{eq-mix-el-proof2})
 must vanish, and thus
 \begin{equation}
   (\axial + \shear)' +
   \left(
     \frac{\mbs{d}}{|\mbs{r}'|}\times
     \nabla_{\mbs{d}'} \bending
   \right)'
   +\bar{\mbs{n}}
   =
   \dot{\mbs{p}}
   +
   \left(
     \frac{\mbs{d}}{|\mbs{r}'|}\times
     \nabla_{\dot{\mbs{d}}} \mbs{\pi}_\perp
   \right)'
   \ ,
  \label{eq-mix-el-proof3}
\end{equation}
with $\shear= (\mbs{\Lambda} \mbs{\eta})\times
\mbs{e}_3$. Next, we consider the second integral
in Eq.~(\ref{eq-mix-el-proof2}), which must also vanish in view of the arbitrariness of
the variation $\delta \mbs{\theta}$. Hence
\begin{equation}
  \mbs{\kappa}\times \bending + \mbs{r}'\times \shear + (\mu\, \mbs{d})'
  = \mbs{w}_\perp\times \mbs{\pi}_\perp\:.
  \label{eq-mix-el-proof4}
\end{equation}
Deriving the relation $0=\bending \cdot \mbs{d}$
with respect to the arc-length we get that $\bending \cdot \mbs{d}' = - \bending'\cdot
\mbs{d}$, and it follows that
\begin{equation}
  \mbs{k}\times \bending = (\mbs{d}\times \mbs{d}')\times \bending = (\bending'\cdot
  \mbs{d})\mbs{d} .
  \label{eq-mix-el-proof6}
\end{equation}
Then, projecting both sides of Eq.~(\ref{eq-mix-el-proof4}) in the direction
of~$\mbs{d}$ and using Eq.~(\ref{eq-mix-el-proof6}) we obtain:
\begin{equation}
    \mbs{m}' \cdot\mbs{d} = \dot{\mbs{\pi}}_\perp\cdot \mbs{d}\ ,
  \label{eq-mix-el-proof7}
\end{equation}
with $\mbs{m}=\bending + \mu\, \mbs{d}$. Using Eq.~\eqref{eq-mix-el-cons},
a simple manipulation of Eq.~\eqref{eq-mix-el-proof7} yields
\begin{equation}
  (\bending + \storsion \mbs{d})'\cdot \mbs{d}
  + \bartorsion
  =
  \dot{\mbs{\pi}}_\perp\cdot \mbs{d}\ ,
  \label{eq-mix-el-proof71}
\end{equation}
which coincides with Eq.~\eqref{eq-mix-el2}. Projecting Eq.~(\ref{eq-mix-el-proof6}) onto the plane
orthogonal to $\mbs{d}$ we get
\begin{equation}
  \mbs{r}'\times \shear + \mu\, \mbs{d}' = 0\ ,
  \label{eq-mix-el-proof8}
\end{equation}
from where we can obtain the closed form expression for the symbol $\shear$:
\begin{equation}
  \shear
  = 
  \mu \frac{\mbs{d}}{|\mbs{r}'|}\times \mbs{d}'
  = 
  \frac{\mbs{d}}{|\mbs{r}'|}\times \nabla_{\mbs{d}'}(\mu\,\mbs{d})\ .
  \label{eq-mix-el-proof9}
\end{equation}
Replacing the value
of $\shear$ in Eq.~(\ref{eq-mix-el-proof3}) we arrive at the final expression:
\begin{equation}
   \left(
   \axial+
     \frac{\mbs{d}}{|\mbs{r}'|}\times
     \nabla_{\mbs{d}'} \bending
   \right)'
   +\bar{\mbs{n}}
   =
   \dot{\mbs{p}}
   +
   \left(
     \frac{\mbs{d}}{|\mbs{r}'|}\times
     \nabla_{\dot{\mbs{d}}} \mbs{\pi}_\perp
   \right)'
   \ .
  \label{eq-mix-el-proof10}
\end{equation}

\end{proof}

Interestingly, when formulating this variational principle for a transversely
isotropic rod several simplifications follow. First, by definition, the stored
energy function depends on $\mbs{K}$ only through its modulus, so the
potential energy is not a function of the section orientation. Then, again
by definition of isotropy, the inertia $\mbs{I}_\perp$ is a multiple of
the identity in $\mathrm{span}(\mbs{E}_1,\mbs{E}_2)$, and thus the kinetic energy does not
depend on the section orientation either. As a result, for the transversely
isotropic case, the fields $\mbs{r},\psi$ can be found by applying Hamilton's
principle to the action~(\ref{eq-mixed-action})
\emph{without solving for $\mbs{\Lambda},\mbs{\eta},\mu$.} In a way,
the problem decouples: the solution for $\mbs{r},\psi$ is independent
of the rest of unknowns, which can be calculated after the expressions for
the former are found. In many practical applications, the latter might even
be superfluous.

It bears emphasis that the action~(\ref{eq-mixed-action}) is not just the action of one of the two
models presented, respectively, in Sections~\ref{sec-canonical} and~\ref{sec-ti}, with some added
terms. The key for this new variational formulation is the reparametrization of the strain
measures and the approximation of the kinetic energy. In fact, it is easy to verify
that the functional of Section~\ref{sec-canonical} for general rods does not simplify
when the model is transversely isotropic.
%

\section{Summary}
\label{sec-summary}
This article presents a systematic analysis of Kirchhoff rods for general and transversely isotropic
cross sections. The framework is that of variational analysis that, elegantly, provides the
governing and boundary, conditions of each model, and helps to identify potential pitfalls in their
geometric setting.

The first result is a variational principle for general rods
that can deform axially, in bending, and in torsion, but not in
shear. This is a minor modification of the well-known
counterpart of inextensible Kirchhoff rods and provides a clear
illustration of the formulation of structural models by means
of constrained variational principles.

A very common motivation for mechanicians is the development
of a constraint-free theory for Kirchhoff rods. We discuss that
this is just possible, in the quasistatic case, for transversely
isotropic beams, and we provide the corresponding variational
principle. For dynamic problems we argue that it is not possible
to find an exact constraint-free theory, not even for transversely
isotropic rods. For the latter case we postulate \emph{approximate}
models that admit a variational principle free of any constraint.

The previous results provide two alternative variational
principles: one valid for any Kirchhoff rod -- but constrained -- and
a second, constraint-free, valid for transversely isotropic rods. With
a view in numerical methods, we propose a third route: a variational
principle that is valid for general Kirchhoff rods but simplifies
in the  transversely isotropic case, so that the constraints
can be solved pointwise in closed form, as long as the kinetic
energy is approximated.

\ifsage
\begin{funding}
  \else
  \section{Acknowledgements}
  \fi

  I.R. was funded by project DPI2017-92526-EXP from the Spanish Ministry of Economy, Industry and Competitiveness.
\ifsage
\end{funding}
\fi
\bibliographystyle{SageV}
\bibliography{bib}

\end{document}